\documentclass[a4paper,UKenglish,autoref]{lipics-v2019}
\usepackage{times,epsfig,epsf,psfrag,array,amsmath,amssymb,verbatim}
\usepackage{url,comment,enumerate,graphicx,graphics,wrapfig}
\usepackage[usenames,dvipsnames,svgnames,table]{xcolor}
\usepackage[ruled,linesnumbered,vlined]{algorithm2e}
\usepackage{multirow,amsthm}

\usepackage{pdfpages,framed,tcolorbox}
\usepackage{tabularx}
\usepackage{pdflscape,latexsym,bm,array,caption,textcomp}

\newcolumntype{a}{>{\centering\arraybackslash}p{4.9em}}
\newcolumntype{q}{>{\centering\arraybackslash}p{7em}}
\newcolumntype{w}{>{\centering\arraybackslash}p{10em}}
\newcolumntype{C}{>{\centering\arraybackslash}p{6em}}

\renewcommand{\O}{\widetilde O}

\def\l1{\mathbf{L_1}}

\long\def\comment #1\commentend{}
\def\inline#1:{\par\vskip 7pt\noindent{\bf #1:}\hskip 10pt}

\title{Diameter Spanners,  Eccentricity Spanners, and Approximating Extremal Distances}

\author{Keerti Choudhary}
{Weizmann Institute of Science, Israel. E-mail: keerti.choudhary@weizmann.ac.il}{}{}{}
\author{Omer Gold}
{Tel Aviv University, Israel. E-mail: omergold@post.tau.ac.il}{}{}{}

\authorrunning{K. Choudhary and O. Gold}

\Copyright{Keerti Choudhary and Omer Gold}

\ccsdesc[500]{Theory of computation~Sparsification and spanners}
\ccsdesc[500]{Theory of computation~Dynamic graph algorithms}

\keywords{Diameter, Eccentricity, Spanner, Dynamic-graph-algorithms, Fault-tolerant.}

\nolinenumbers 

\hideLIPIcs  

\EventEditors{John Q. Open and Joan R. Access}
\EventNoEds{2}
\EventLongTitle{42nd Conference on Very Important Topics (CVIT 2019)}
\EventShortTitle{CVIT 2019}
\EventAcronym{CVIT}
\EventYear{2019}
\EventDate{April 24--27, 2019}
\EventLocation{Little Whinging, United Kingdom}
\EventLogo{}
\SeriesVolume{42}
\ArticleNo{23}

\newcommand{\inbfs}{\textsc{in-bfs}}
\newcommand{\outbfs}{\textsc{out-bfs}}
\newcommand{\nin}{N^{in}}
\newcommand{\nout}{N^{out}}

\newcommand{\A}{{\cal A}}
\newcommand{\B}{{\cal B}}

\newcommand{\dia}{\mathrm{diam}}
\newcommand{\rad}{\mathrm{rad}}
\newcommand{\inecc}{\textsc{i}\mathrm{n}\textsc{e}\mathrm{cc}}
\newcommand{\outecc}{\textsc{o}\mathrm{ut}\textsc{e}\mathrm{cc}}
\newcommand{\depth}{\textsc{depth}}
\newcommand{\inv}{N^{in}}

\newcommand{\maxD}{D_{{max}}}
\newcommand{\thD}{D_{0}}
\renewcommand{\O}{\widetilde O}
\renewcommand{\P}{{\cal P}}

\begin{document}
\maketitle

\newcommand\blfootnote[1]{%
  \begingroup
  \renewcommand\thefootnote{}\footnote{#1}%
  \addtocounter{footnote}{-1}%
  \endgroup
}
\blfootnote{A graph in this paper always refer to directed graph. 
Notations $n,m$ are used to denote the size of the vertex-set and the edge-set of a graph respectively.}
\blfootnote{$\O(\cdot)$ hides poly-logarithmic factors.}

\begin{abstract}
The diameter of a graph is one if its most important parameters, being used in many real-word applications.
In particular, the diameter dictates how fast information can spread throughout data and communication networks.  
Thus, it is a natural question to ask how much can we sparsify a graph and still guarantee that its diameter remains preserved within an approximation $t$. This property is captured by the notion of extremal-distance spanners. Given a graph $G=(V,E)$, a subgraph $H=(V,E_H)$ is defined to be a {\em $t$-diameter spanner} if the diameter of $H$ is at most $t$ times the diameter of $G$.

We show that for any $n$-vertex and $m$-edges directed graph $G$, we can compute a sparse subgraph $H$ that is a $(1.5)$-diameter spanner of $G$, such that $H$ contains at most $\O(n^{1.5})$ edges. We also show that the stretch factor cannot be improved to $(1.5-\epsilon)$. 
For a graph whose diameter is bounded by some constant, we show the existence of $\frac{5}{3}$-diameter spanner that contains at most $\O(n^\frac{4}{3})$ edges.
We also show that this bound is tight.

Additionally, we present other types of extremal-distance spanners, such as $2$-eccentricity spanners and $2$-radius spanners, both contain only $\O(n)$ edges and are computable in $\O(m)$ time.

Finally, we study extremal-distance spanners in the dynamic and fault-tolerant settings. An interesting implication of our work is the {\em first $\O(m)$-time algorithm} for computing $2$-approximation of vertex eccentricities in general directed weighted graphs. Backurs et al. [STOC 2018] gave an $\O(m\sqrt{n})$ time algorithm for this problem, and also showed that no $O(n^{2-o(1)})$ time algorithm can achieve an approximation factor better than $2$ for graph eccentricities, unless SETH fails; this shows that our approximation factor is essentially tight.
\end{abstract}

\section{Introduction}


In this paper, we initiate the study of  {\em Extremal Distance Spanner}.
The notion of spanners (also known as distance spanner) was first introduced and studied in 
\cite{Awerbuch85,PS89,PelegU89a}. 
A {\em spanner} (also known as distance spanner) of a graph $G=(V,E)$ is a sparse subgraph 
$H=(V,E_H)$ that approximately preserves pair-wise distances of the underlying graph $G$.  
Besides being theoretically interesting, they are known to have numerous applications in different areas of computer 
science such as distributed systems, communication networks and efficient routing schemes~\cite{Cowen01, 
Cowen04, PU89, RTZ02, TZ01, GS11, ALWWF03, GZ10}, motion planning~\cite{DB14, CSMOC15}, approximating 
shortest paths~\cite{Cohen98, Cohen00, Elkin01} and distance oracles~\cite{BS06, TZ05}. 

It is known that for any  integer $k \geq 1$, there exists for an undirected graph with $n$ vertices 
a $O(n^{1+1/k})$ edges spanner with multiplicative stretch $2k-1$.
The works of~\cite{RodittyTZ05,BS03} provided efficient constructions of such spanners.
It is also widely believed that this size-stretch trade-off is tight. 
Assuming the widely believed Erd{\"o}s girth conjecture~\cite{Erdo64}, this size-stretch trade-off is tight.
Other fascinating works have studied spanners for undirected graphs with additive 
stretch~\cite{BKMP10, Chechik2013, ACIM99, EP04,ThorupZ06}, spanners for different distance 
metrics~\cite{Cowen04, RTZ02, ALWWF03}, and so on. 
It was shown in \cite{EP04} that for any $\epsilon > 0$ and any $\kappa \geq 1$, there 
exists a $(1 + \epsilon, \beta)$-spanner for $n$-vertex unweighted graphs with 
$O(\beta n^{1+1/\kappa})$ edges, where $\beta = O(\log \frac{\kappa}{\epsilon})^{\log\kappa}$. 

Unfortunately, the landscape of distance spanners in the directed setting is far less understood.
This is because we cannot have sparse spanners for general directed graphs.
Even when underlying graph is strongly-connected, there exists graphs with $\Omega(n^2)$
edges such that excluding even a single edge from the graph results in a distance-spanner with 
stretch as high as diameter. In such a scenario, for directed graphs, a natural direction to study is 
construction of sparse subgraphs that approximately preserves the graph diameter.\vspace{1mm}

This brings us to the following central question.

\begin{question}
Given a directed graph $G = (V, E)$ and a ``stretch factor'' $t$, can we construct a sparse subgraph 
$H$ of $G$ such that the distance between any two vertices in $H$ is bounded by $t$
times the maximum distance in $G$? 
\end{question}

We define such graphs as {\em $t$-diameter spanners} as they essentially preserve the diameter up to a multiplicative factor $t$. 
We also consider the following related question of {\em $t$-eccentricity spanner}.

\begin{question}
Given a directed graph $G = (V, E)$ and a ``stretch factor'' $t$, can we construct a sparse subgraph 
$H$ of $G$ such that the eccentricity of each vertex $v$ in $H$ is at most  $t$
times the eccentricity of $v$ in $G$? 
\end{question}

We study the unexplored terrain of extremal-distance spanners for directed graphs. 
First, it is not clear if there exists a $t$-diameter spanner and $t$-eccentricity spanner
for directed graphs for small values of $t$~(typically close to $1$). 
Next, suppose there exists a $t$-diameter/eccentricity spanner for 
some parameter $t$, how can we construct such spanners efficiently?
Finally, we address the questions of maintaining  these spanners in the dynamic setting and the fault-tolerant setting, 
where the underlying graph changes with time. 

We believe that extremal-distance spanners are interesting mathematical objects in their own right.
Nevertheless, such a sparsification of graphs indeed suffices for many of the original applications of the well-studied standard graph spanners,
such as in communication networks, facility location problem, routing, etc.
In particular, diameter spanners with a sparse set of edges are good candidates for backbone networks~\cite{GZ10}. 
Our study of extremal-distance spanners has many additional implications that we present in the next subsection.

\subsection{Our Contributions}

In the following subsections, we present our results in detail.

\subsubsection{Diameter Spanners}

We show sparse diameter spanner constructions with various trade-offs between the size (number of edges) 
of the spanner and its {\em stretch} factor $t$, and provide efficient algorithms to construct such spanners.

We provide efficient construction of $1.5$-diameter spanners, 
and also show that our $1.5$-stretch diameter spanner construction is 
essentially tight for graphs whose diameter is bounded by $O(n^{1/4})$.

\begin{theorem}\label{theorem:3/2-diam-up-lb}
(a) There exists a Las Vegas algorithm that for any unweighted directed graph $G$, computes
a $1.5$-diameter spanner $H$ of $G$ with at most $O(n^{3/2}\sqrt{\log n})$ edges.
The computation time of $H$ is $\O(m\sqrt{n})$ with high probability.
If $G$ is edge-weighted, then $H$ satisfies the condition that $\dia(H)\leq 1.5 ~\dia(G)+W$,
where $W$ is an upper bound on the weight of edges in $G$.

(b) For every $n$ and every $D\leq (n^{1/4})$, there exists an unweighted directed graph $G$ 
with $\Theta(n)$ vertices and diameter $D$, such that any subgraph $H$ of $G$ that satisfies
$\dia(H)$ is strictly less than $1.5\hspace{1mm}\dia(G)-1$ contains at least $\Omega(n^{1.5})$ edges.
\end{theorem}

For the scenario when $D=o(\sqrt{n})$, we provide a construction of 
$5/3$-diameter spanners that are sparser than the $1.5$-diameter spanners.
We also show that our $5/3$-stretch spanner construction is tight, for graphs whose
diameter is bounded by $o(n^{1/11})$.

\begin{theorem}\label{theorem:5/3-diam-up-lb}
(a) There exists a Las Vegas algorithm that for any directed graph $G$ having diameter~$D~\leq \sqrt n$, computes 
a $5/3$-diameter spanner $H$ of $G$ that contains at most $O(n^\frac{4}{3}D^{\frac{1}{3}}\log^{\frac{2}{3}} n)$ edges.
The computation time of $H$ is $\widetilde O(mn^{2/3}D^{2/3})$ with high probability%
\footnote{Though the computation time of $H$ is a function of $D$, the algorithm does not 
need to apriori know the value $D$.}.

(b) For every $n$ and every $D\leq (n^{1/11})$, there exists a unweighted directed graph $G$ 
with $\Theta(n)$ vertices and diameter $D$, such that any subgraph $H$ of $G$ for which
$\dia(H)$ is strictly less than $(5/3\hspace{1mm}\dia(G)-1)$ contains at least $\Omega(n^{4/3}D^{1/3})$ edges.
\end{theorem}

We also show that for any directed graph $G$ we can either 
(i) compute a diameter spanner with arbitrarily low stretch, or
(ii) compute a diameter spanner with arbitrarily low size.

\begin{theorem}
For any arbitrarily small fractions $\epsilon,\delta>0$, and any given directed graph $G$, 
in~$\widetilde O(mn^{1-\epsilon})$ expected time, at least one of the following subgraphs
can be computed. \\[1mm]
(i) a $(1+\delta)$-diameter spanner of $G$ containing at most $O(n^{2-\epsilon}\sqrt{\log n})$ edges.\\[1mm]
(ii) a $(2-\delta)$-diameter spanner of $G$ containing at most $O(n^{1+\epsilon}\sqrt{\log n})$ edges.
\end{theorem}

In Theorem~\ref{theorem:3/2-diam-up-lb} and Theorem~\ref{theorem:5/3-diam-up-lb},
we show a lower bound on the number of edges in diameter spanners of stretch respectively $3/2$ and $5/3$,
for graphs with low diameter. A natural question to ask here is if for graphs with large diameter it is possible to
obtain diameter spanners with low stretch~(ideally $1+o(1)$) that are also sparse~(ideally having $\O(n)$ edges)
in nature. The next theorem positively answers this question.

\begin{theorem}
For any directed graph $G=(V,E)$ satisfying $\dia(G)=\omega(n^{5/6})$, we can compute a subgraph $H=(V,E')$ 
with $O(n\log^2n)$ edges satisfying $\dia(H)\leq (1+o(1))\dia(G)$.
\end{theorem}


\subsubsection{Dynamic Maintenance of Diameter Spanners}\label{sec:dynamic-spanners}
We obtain the following dynamic algorithm for maintaining a $(1.5+\epsilon)$-diameter spanner
incrementally, as well as decrementally.

\begin{theorem}
For any $\epsilon>0$ and $n$-vertex directed graph, there exists an incremental (and decremental) algorithm that
maintains a $(1.5+\epsilon)$-diameter spanner that consists at most $O(n^{1.5}\sqrt{\log n})$ edges.
The expected amortized update time of the algorithm is $\O(\sqrt{n}/\epsilon^{2}~D_0)$
for the incremental setting and $\O(\sqrt{n}/\epsilon~D_0)$ for the decremental setting,	
where $D_0$ denotes an upper bound on the diameter of the graph throughout the run of the algorithm.
\end{theorem}

For graphs whose diameter remains bounded by $o(\sqrt{n})$, we provide incremental (and decremental)
algorithms for maintaining $(5/3+\epsilon)$-diameter spanners. (In the dynamic setting as well these
spanners are sparser than $1.5$-stretch diameter spanners for graphs whose diameter is at most $o(\sqrt{n})$).

\begin{theorem}
For any $\epsilon>0$ and $n$-vertex directed graph, there exists an incremental (and decremental) algorithm that
maintains a $(5/3+\epsilon)$-diameter spanner that consists at most 
$\O(n^{4/3}D_0^{1/3})$ edges, where $D_0~(\leq \sqrt{n})$ denotes an upper bound on the 
graph diameter throughout the run of the algorithm.
The expected amortized update time of the algorithm is $\O(\epsilon^{-2} n^{2/3}D_0^{2/3})$
for the incremental setting, and $\O(\epsilon^{-1} n^{2/3}D_0^{2/3})$ for the decremental setting.
\end{theorem}

An interesting immediate corollary of our dynamic maintenance of diameter spanners is an incremental (and decremental) algorithm that
maintains a $(1.5+\epsilon)$-approximation of the graph's diameter,
whose expected total update time is $\O(\epsilon^{-2}m\sqrt{n}D_0)$ 
for the incremental setting, and $\O(\epsilon^{-1}m\sqrt{n}D_0)$ for the decremental setting,	
where $D_0$ denotes an upper bound on the diameter of graph throughout the run of the algorithm.
There is a very recent independent work by Ancona {\em et al.}~\cite{arxiv:diam} on dynamically maintaining the diameter value of a graph, using related techniques.
In particular, they give a $(1 + \epsilon)$-approximation algorithm with $\O(\epsilon^{-2}nm)$ total update time for the incremental setting, and a $(1.5 + \epsilon)$-approximation algorithm for incremental (resp., decremental) setting, whose total update time is $\O(\epsilon^{-2} m\sqrt{nD_0})$ with high probability.
\footnote{The authors of~\cite{arxiv:diam} informed us that they had a solution to the decremental problem since 2015 (private communication).}

\subsubsection{Eccentricity Spanners}
Given a graph $G=(V,E)$, we say that a subgraph $H=(V,E'\subseteq E)$of $G$ is a $t$-eccentricity spanner
of $G$ if the eccentricity of any vertex $x$ in $H$ is at most $t$ times its eccentricity in $G$.
Similarly, $H$ is said to be a $t$-radius spanner, if the radius of $H$ is at most $t$ times the radius of graph $G$.

We obtain a construction for $2$-eccentricity spanner with $\O(n)$ edges that are computable in just $\O(m)$ time;
also we show that there exists graphs whose $(2-\epsilon)$-eccentricity spanner, for any $\epsilon>0$,
contains $\Omega(n^2/R^2)$ number of edges, where $R$ is the radius of the graph.

\begin{theorem}\label{theorem:2-ecc-up-lb}
(a) There exists a Las Vegas algorithm that for any directed weighted graph $G=(V,E)$
computes in $O(m\log^2 n)$ expected time a $2$-eccentricity spanner (which is also 
a $2$-radius spanner) containing at most $O(n\log^2n)$ edges.

(b) For every $n$ and every $R$, there exists a unweighted directed graph $G$ with $\Theta(n)$ vertices and radius $R$, 
such that any subgraph $H$ of $G$ that is a $(2-\epsilon)$-eccentricity spanner
of $H$, for any $\epsilon>0$, contains at least $\Omega(n^2/R^2)$ edges.
%
\end{theorem}

\subparagraph*{Implications: First near optimal time $2$-approximation of vertex eccentricities.}
For the problem of computing exact eccentricities in weighted directed graphs the only known solution 
is solving the all-pair-shortest-path problem and it takes $\Omega(mn)$ time. Backurs et al.~\cite{Backurs18} 
showed that for any directed weighted graph there exists an algorithm for computing 2-approximation 
of eccentricities in $\tilde O(m\sqrt{n})$ time. They also showed that for any $\delta>0$, there exists 
an $\O(m/\delta)$ time algorithm for computing a $2+\delta$ approximation of graph eccentricities.

We obtain (as a product of our $2$-eccentricity spanner) the first $\O(m)$ time algorithm for 
computing $2$-approximation of eccentricities.

\begin{theorem}
For any directed weighted graph $G=(V,E)$ with $n$ vertices and $m$ edges, we can compute
in $O(m\log^2 n)$ expected time a $2$-approximation of eccentricities of vertices in $G$.
 \end{theorem}
 
Our result is essentially {\bf \em tight}.
The approximation factor of $2$ cannot be improved since Backurs et al. showed in their paper that
unless SETH fails no $O(n^{2-o(1)})$ time algorithm can achieve an approximation 
factor better than $2$ for graph eccentricities~\cite{Backurs18}.
Also the computation time of our algorithm is almost optimal as we need $\Omega(m)$ time to even scan the edges of the graph.

\subsubsection{Dynamic Maintenance of Eccentricity Spanners}

We obtain incremental and decremental algorithm for maintaining $(2+\epsilon)$-eccentricity spanner.

\begin{theorem}
For any $\epsilon>0$, there exists an incremental (and decremental) algorithm that
maintains for an $n$-vertex directed graph a $(2+\epsilon)$-eccentricity spanner 
(which is also a $(2+\epsilon)$-radius spanner) containing at most $O(n^{3/2}\sqrt{\log n})$ edges.
The expected amortized update time is $\O(\sqrt{n}/\epsilon^{3}~D_0 )$ 
for the incremental setting and $\O(\sqrt{n}/\epsilon~D_0)$ for the decremental setting,	
where, $D_0 $ denotes an upper bound on the diameter of the graph throughout the run of the algorithm.
\end{theorem}

\subparagraph*{Implication: Dynamic maintenance of $2$-approximate eccentricities.}
The above stated dynamic algorithm for $(2+\epsilon)$-eccentricity spanner also imply 
a same time bound algorithm for maintaining a $(2+\epsilon)$-approximation of vertex eccentricities. 

\begin{theorem}
For any $\epsilon>0$, there exists an incremental (and decremental) algorithm that
maintains for an $n$-vertex directed graph a $(2+\epsilon)$-approximation of graph eccentricities.
The expected amortized update time is $\O(\sqrt{n}/\epsilon^{3}~D_0 )$ 
for the incremental setting and $\O(\sqrt{n}/\epsilon~D_0)$ for the decremental setting,	
where, $D_0 $ denotes an upper bound on the diameter of the graph throughout the run of the algorithm.
\end{theorem}

Additional dynamic algorithms with different approximation factors and runtime bounds for maintaining vertex eccentricities, were very recently given by Ancona {\em et al.}~\cite{arxiv:diam}; for more details, see~\cite{arxiv:diam}.

\subsection{More Related Work}\label{related-work}
The girth conjecture of Erd{\"o}s~\cite{Erdo64} implies that there are undirected graphs on $n$ 
vertices, for which any $(2k-1)$-spanner will require $\Omega(n^{1+1/k})$ edges.
This conjecture has been proved for $k=1,2,3$, and $5$, and is widely believed to be true for any integer $k$.
Thus, assuming the girth conjecture, one can not expected for a better size-stretch trade-offs.

Alth{\"o}fer {\em et al.}~\cite{ADDJS93} were the first to show that any undirected weighted graph 
with $n$ vertices has a $(2k-1)$-spanner of size $O(n^{1+1/k})$.
The lower bound mentioned above implies that the $O(n^{1+1/k})$ size-bound of this spanner is essentially optimal.
Alth{\"o}fer {\em et al.}~\cite{ADDJS93} gave an algorithm to compute such a spanner,
and subsequently, a long line of works have studied the question of how fast can we compute such a spanner,
until Baswana and Sen~\cite{BS03} gave a linear-time algorithm.

A $c$-additive spanner of an undirected graph $G$ is a subgraph $H$ that preserves distances up to an 
additive constant $c$. That is, for any pair of nodes $u,v$ in $G$ it holds that $d_H(v,u) < d_G(v,u) + c$. 
This type of spanners were also extensively studied~\cite{BKMP10, Chechik2013, ACIM99, EP04}.
The sparsest additive spanner known is a 6-additive spanner of size $O(n^{4/3})$ that was given by
Baswana, Kavitha, Mehlhorn, and Pettie~\cite{BKMP10}. It was only recently that Abboud and 
Bodwin~\cite{AB16} proved that the $O(n^{4/3})$ additive spanner bound is tight, for any additive constant $c$.

Since for directed graph distance spanners are impossible, the {\em roundtrip distance} 
metric was proposed. The roundtrip-distance between two vertices $u$ and $v$ is 
the distance from $v$ to $u$ plus the distance from $u$ to $v$. 
Roditty, Thorup, and Zwick~\cite{RTZ02} presented the notion of {\em roundtrip spanners} for directed graphs.
A roundtrip spanner of a directed graph $G$ is a sparse subgraph $H$ that approximately preserve the {\em roundtrip} distance 
between each pair of nodes $v$ and $u$. They showed that any directed graph has roundtrip spanners, and gave efficient algorithms to construct such spanners.

The question of finding the sparsest spanner of a given graph was shown to be NP-Hard by Peleg and Sch{\"a}ffer~\cite{PS89}, in the same work that graph spanner notion was introduced by Peleg and Sch{\"a}ffer~\cite{PS89}.

Diameter spanners were mentioned by Elkin and Peleg~\cite{ElkinP01, ElkinPeleg01},
but in the context of approximation algorithms for finding the sparsest diameter spanner (which is NP-Hard).
To the best of our knowledge, there is no work that showed the existence of sparse diameter spanners with stretch less than $2$,
for directed graphs.

\section{Preliminaries}
\label{section:prelim}
Given a directed graph $G=(V,E)$ on $n=|V|$ vertices and $m=|E|$ edges, 
the following notations will be used throughout the paper.

\begin{itemize}
\itemsep0em
\item $\pi_G(v,u)$:~ the shortest path from vertex $v$ to vertex $u$ in graph $G$.
\item $d_G(v,u)$:~ the length of the shortest path from vertex $v$ to vertex $u$ in graph $G$.
We sometimes denote it by $d(v,u)$, when the context is clear.
\item $d_G(A,u)$:~ $\min_{v\in A}d_G(v,u)$.
\item $\dia(G)$:~ the diameter of graph $G$, that is, $\max_{u,v\in V}{d_G(u,v)}$.
\item $\outbfs(s)$ ($\outbfs(S)$):~ an outgoing breadth-first-search (BFS) tree rooted at $s$ (supernode $S$).
\item $\inbfs(s)$ ($\inbfs(S)$):~ an incoming breadth-first-search (BFS) tree rooted at $s$ (supernode $S$).
\item $\inecc(s)$:~ the depth of tree $\inbfs(s)$.
\item $\outecc(s)$:~ the depth of tree $\outbfs(s)$.
\item $\rad(G)$:~ the radius of graph $G$, that is, $\min_{v\in V}\outecc(v)$.
\item $\outbfs(s, d)$:~ the tree obtained from $\outbfs(s)$ by truncating it at depth $d$.
\item $\inbfs(s, d)$:~ the tree obtained from $\inbfs(s)$ by truncating it at depth $d$.
\item $\nout(x,\ell)$:~ the $\ell$ closest outgoing vertices of $x$, where ties are broken arbitrarily.
\item $\nin(x,\ell)$:~ the $\ell$ closest incoming vertices of $x$, where ties are broken arbitrarily.
\item $\depth(v,T)$:~ the depth of vertex $v$ in the rooted tree $T$.
\item $\depth(T,W)$:~ $\max_{w\in W}\depth(w,T)$.
\item $\P(W)$:~ the power set of $W$.
\end{itemize}

In all the above defined notations, when dealing with dynamic graphs, we use subscript 
$t$ to indicate the time-stamps.

Throughout the paper we assume the graph $G$ is strongly connected, 
as otherwise  the diameter of $G$ is $\infty$, and even an empty subgraph of $G$ preserves its diameter.

We first formally define the notion of the $k$-diameter spanners that is used in the paper.
\begin{definition}[Diameter-spanner]
Given a directed graph $G=(V,E)$, a subgraph $H=(V,E'\subseteq E)$
is said to be $k$-diameter spanner of $G$ if $\dia(H)\leq \lceil k\cdot\dia(G)\rceil$.
\end{definition}

Next we introduce the notion of $\langle h_1, h_2 \rangle$-dominating-set-pair
which are a generalization of traditional $h$-dominating sets~\cite{haynes1998fundamentals,henning2014total}. 

\begin{definition}[Dominating-set-pair]
For a directed graph $G=(V,E)$, and a set-pair $(S_1, S_2)$ satisfying $S_1,S_2\subseteq V$, we say that 
$(S_1,S_2)$ is $\langle h_1,h_2\rangle$-dominating with size-bound  $\langle n_1,n_2\rangle$, 
if $|S_1|=O(n_1)$, $|S_2|=O(n_2)$, and either
(1) for each $x\in V$, $d_G(S_1,x)\leq h_1$, or 
(2) for each $x\in V$, $d_G(x,S_2)\leq h_2$. 

Here, $S_1$ is said to be $h_1$-out-dominating if it satisfies condition~1,
and $S_2$ is said to be $h_2$-in-dominating if it satisfies condition~2.
\end{definition}

We point here that it was shown by Cairo, Grossi and Rizzi~\cite{CairoGR16} that for any $n$-vertex
undirected graph with diameter $D$, there is a $\tiny{\frac{2^k-1}{2^k}}(D+1)$-dominating-set of size
$\O(n^{1+1/k})$. Using the $\frac{2^k-1}{2^k}(D+1)$-dominating-set 
they obtain a hierarchy of diameter and radius approximation algorithms. 
However, the construction extends to directed graph only when $k=1$.
In fact, we can prove that these bounds are unachievable for directed graphs whenever $k>1$.
For completeness, we revisit the construction for $k=1$ in Section~\ref{section:main}.

We below state few results that will useful in our construction.

\begin{lemma}
Let $G = (V, E)$ be an $n$-vertex directed graph. 
Let $n_p,n_q\geq 1$ be integers satisfying $n_pn_q= 8n\log n$, 
and let $S$ be a uniformly random subset of $V$ of size $n_p$.
Then with a high probability, $S$ has non-empty intersection with $\nin(v,n_q)$
and $\nout(v,n_q)$, for each $v\in V$.
\footnote{If the graph $G$ is undergoing either edge insertions, or edge deletions,
then with high probability the relation holds for each of the $O(n^2)$ instances of $G$.}
\label{lemma:probability}
\end{lemma}	

In order to dynamically maintain diameter-spanners, we will use the following 
result by Even and Shiloach~\cite{ES:81} on maintaining single-source shortest-path-trees.
Even and Shiloach gave the algorithm for maintaining shortest path tree
in the decremental setting, and their algorithm can be easily adapted to work in the incremental 
setting as well.
 
\begin{theorem}[ES-tree~\cite{ES:81}] There is a decremental (incremental) algorithm for
maintaining the first $k$ levels of a single-source shortest-path-tree, in a directed
or undirected graph, whose total running time, over all deletions (insertions), is
$O(km)$, where $m$ is the initial (final) number of edges in the graph.
\label{theorem:ES}
\end{theorem}

\section{Static and Dynamic Maintenance of 1.5-Diameter Spanners}
\label{section:main}

Our main idea for computing sparse diameter spanner comes
from the recent line of works~\cite{RodittyW:13,Backurs18,ChechikLRSTW:14,HolzerPRW:14,CairoGR16} 
on approximating diameter in directed or undirected graphs.
Let $S_1$ be a uniformly random subset of $V$ of size $\sqrt{n\log n}$.
We take $a$ to be the vertex of the maximum depth in $\outbfs(S_1)$.
Also, $S_2$ is set to $\nin(a,\sqrt{n\log n})$.
By Lemma~\ref{lemma:probability}, with high probability, the set $\nin(a,\sqrt{n\log n})$ contains a vertex of 
$S_1$, if not, we can re-sample $S_1$, and compute $a,S_2$ again.
For convenience, throughout this paper, we refer to this constructed set-pair $(S_1,S_2)$ as a {\em valid set-pair}.

Using the idea of $(D+1)/2$-dominating-set construction of Cairo~et ~al.~\cite{CairoGR16}
for undirected graphs,
we show that a valid-set-pair is $\langle \lfloor{pD}\rfloor,\lceil{qD}\rceil \rangle$-dominating 
for any fractions $p,q>0$ satisfying $p+q= 1$, where $D$ is the diameter of the input graph.
If depth of $\outbfs(S_1)$ is at most $ \lfloor{pD}\rfloor$, then $S_1$ is trivially
$\lfloor{pD}\rfloor$-out-dominating. So let us consider the case that depth of $\outbfs(S_1)$ is greater than 
$\lfloor{pD}\rfloor$. Observe that in such a case $\inbfs(a,\lfloor{pD}\rfloor)$ does not contain
any vertex of $S_1$, and so $|\inbfs(a, \lfloor{pD}\rfloor)|=O(\sqrt{n\log n})$ 
because $S_1$ intersects the set $\nin(a,\sqrt{n\log n})$, whereas, $\inbfs(a, \lfloor{pD}\rfloor)\cap S_1$ 
is empty. Since $\inbfs(a, \lfloor{pD}\rfloor)$ is a strict subset of $S_2:=\nin(a,\sqrt{n\log n})$, 
we have that $\depth(\inbfs(S_2))$ is at most $\lceil qD\rceil$ if $G$ is unweighted,
and at most $W+qD$ is $G$ is edge-weighted with weights in range $[1,W]$.

Let $H$ denote the subgraph of $G$ which is union of $\inbfs(s)$ and $\outbfs(s)$, for $s\in S_1\cup S_2$.
Graph $H$ contains at most $O(n^{3/2}\sqrt{\log n})$ edges since $|S_1\cup S_2|=O(\sqrt{n\log n})$.
Observe that computation of graph $H$ takes $O(m\sqrt{n\log n})$ expected time (recall computation of $(S_1,S_2)$
is randomized).
If $S_1$ is $\lfloor D/2\rfloor$-out-dominating, then $H$ is a $1.5$-diameter spanner.
Indeed, for any $x,y\in V$, there is an $s\in S_1$ satisfying $d_H(s,y)=d_G(s,y)\leq D/2$,
and $d_H(x,s)=d_G(x,s)\leq D$, which implies $d_H(x,y)\leq 1.5 D$. In the case when $S_1$ is not
$\lfloor D/2\rfloor$-out-dominating, then as shown $S_2$ is $\lceil 0.5D\rceil$-in-dominating if $G$ is unweighted,
and $(W+D/2)$-in-dominating is $G$ is edge-weighted, so in this case, respectively, $d_H(x,y)$
is bounded by $\lceil 1.5D\rceil$ or $(W+1.5D)$, for every $x,y\in V$. 
We thus have the following theorem.

\begin{theorem}
For any directed unweighted graph $G$ with $n$ vertices and $m$ edges, we can compute
in $O(m\sqrt{n\log n})$ expected time a $1.5$-diameter spanner $H$ of $G$ with at most $O(n^{3/2}\sqrt{\log n})$ 
edges.

Moreover, if $G$ is edge-weighted, then $H$ satisfies the condition that $\dia(H)\leq 1.5 ~\dia(G)+W$,
where $W$ is an upper bound on the maximum edge-weight in $G$.
\label{theorem:static_diam1}
\end{theorem}


The construction of $1.5$-diameter spanners is {\em quite trivial}, however,
their dynamic maintenance is {\em challenging}.
In order to maintain a $1.5$-diameter spanner dynamically, a naive approach would be
to dynamically maintain the valid-sets. We face two obstacles:
(i) the first being dynamic maintenance of a vertex $a$ having maximum depth in tree $\outbfs(S_1)$,
and (ii) the second is dynamically maintaining  the set $\nin(a,n_q)$. 
We see in the next subsection, how to tackle these issues.

\subsection{Dynamic Maintenance of Dominating Sets}

In this subsection, we provide efficient algorithms for maintaining a dominating-set-pair.
We first observe that the static construction of 
dominating-set-pair can be even further generalized as given in the following lemma
(see the Appendix for its proof).

\begin{lemma}
For any integers $n_p,n_q\geq 1$ satisfying $n_p\cdot n_q= 8n\log n$, and any directed graph 
$G=(V,E)$ with $n$ vertices and $m$ edges, in $O(m)$ expected time
we can compute a set-pair $(S_1,S_2)$ of size bound $\langle n_p,n_q\rangle$ 
which is $\langle \lfloor{p~\inecc(a)}\rfloor,\lceil{q~\inecc(a)}\rceil \rangle$-dominating 
for some vertex $a\in V$ and any arbitrary fractions $p,q$ satisfying $p+q= 1$.
\label{lemma:dom-set-rand}
\end{lemma}

Our main approach for dynamically maintaining a dominating-set-pair
is to use the idea of lazy updates. We formalize
this through the following lemma.
\begin{lemma}
Let $G$ be a dynamic graph whose updates are insertions (or deletions) of edges, and $S_1$ be a 
(non-dynamic) subset of $V$ of size $n_p$. Let $t\geq t_0$ be two time instances, and let
$S_2=\nin_{t_0}(a,n_q)$, for some $a\in V$.  Let $\ell_0=\depth_{t_0}(\outbfs(S_1))$ and $\epsilon\in[0,1/2]$ 
be such that $\depth_t(\outbfs(S_1))$ and  $\depth_t(a,\outbfs(S_1))$ 
lie in the range $[\ell_0(1-\epsilon),\ell_0(1+\epsilon)]$.  Then for any $p,q>0$ satisfying $p+q=1$, set-pair
$(S_1,S_2)$ is $\langle \lfloor{(p+2\epsilon)~\inecc_{t}(a)}\rfloor,\lceil{(q+2\epsilon)~\inecc_{t}(a)}\rceil 
\rangle$ dominating at time $t$ if $S_1\cap S_2$ is non-empty, and\\[1mm]
(i) if~ $\depth_{t_0}(a,\outbfs(S_1))\geq (1-\epsilon)\ell_0$, when restricted to edge deletions case.\\[1mm]
(ii) if~ $\depth_{t}(\inbfs(a))\geq (1-\epsilon)\depth_{t_0}(\inbfs(a))$, when restricted to edge insertions case.
\label{lemma:dynamic-dom-set-1}
\end{lemma}
\begin{proof}
Let $\delta$ and $\delta_0$ respectively denote the values $\inecc_{t}(a)$, and $\inecc_{t_0}(a)$.

We first analyse the edge deletions case. If depth of $\outbfs(S_1)$ at the time $t$ is bounded 
by $\lfloor(p+\epsilon)\delta\rfloor$, then $S_1$ is $\lfloor(p+\epsilon)\delta\rfloor $-out-dominating.
So let us assume that $\depth_t(\outbfs(S_1))$ is strictly greater than $(p+\epsilon)\delta $.
Then $(p+\epsilon)\delta  <\depth_t(\outbfs(S_1))\leq (1+\epsilon)\ell_0$. 
So $\ell_0>\big(\frac{p+\epsilon}{1+\epsilon}\big)\delta >p\delta $. Note that $\depth_{t_0}(a,\outbfs(S_1))$ 
is bounded below by $(1-\epsilon)\ell_0>(1-\epsilon)p\delta>(p-\epsilon)\delta$, so,
at time $t_0$, the truncated tree $\inbfs_{t_0}(a,(p-\epsilon)\delta )$ must have empty intersection 
with $S_1$. Since $S_2=\inv_{t_0}(a,n_q)$ intersects $S_1$, at time $t_0$, 
$\inbfs_{t_0}(a,(p-\epsilon)\delta )$ must be contained in $\inv_{t_0}(a,n_q)=S_2$. 
The crucial point to observe is that in the decremental scenario, the set $\inbfs(a,(p-\epsilon)\delta )$ 
can only reduce in size with time. Thus $\inbfs_{t}(a,(p-\epsilon)\delta )\subseteq
\inbfs_{t_0}(a,(p-\epsilon)\delta )\subseteq S_2$. Since $S_2$ contains 
$\inbfs_{t}(a,(p-\epsilon)\delta )$, we have $\depth_t(\inbfs(S_2))\leq \lceil \depth_t(\inbfs(a))-(p-\epsilon)\delta 
\rceil= \lceil \delta -(p-\epsilon)\delta \rceil= \lceil (q+\epsilon)\delta \rceil$.
Thus, if $S_1$ is not $\lfloor(p+\epsilon)\delta\rfloor$-out-dominating, then $S_2$ is 
$\lceil (q+\epsilon)\delta \rceil $-in-dominating set.

We next analyse the edge insertions case.
If $\depth_{t_0}(a,\outbfs(S_1))\leq p\delta$, then $S_1$ is $(p+2\epsilon)\delta$ out-dominating
at time $t_0$ since $\depth_{t_0}(\outbfs(S_1))= \ell_0\leq \depth_t(a,\outbfs(S_1))/(1-\epsilon)$
$\leq \depth_{t_0}(a,\outbfs(S_1))/(1-\epsilon)\leq p\delta /(1-\epsilon)\leq (1+2\epsilon)p\delta 
\leq (p+2\epsilon)\delta $. If depth of $a$ in $\outbfs(S_1)$ at time $t_0$ is greater than $p\delta $, then the 
truncated tree $\inbfs_{t_0}(a,p\delta )$ must have an empty intersection with set $S_1$, however, 
the set $S_2=\nin_{t_0}(a,n_q)$ has a non-empty intersection with $S_1$, thus $\inbfs_{t_0}(a,p\delta )
\subseteq \nin_{t_0}(a,n_q)=S_2$. So $\depth_{t_0}(\inbfs(S_2))\leq \lceil\depth_{t_0}(\inbfs(a))
-p\delta \rceil\leq \lceil \delta/(1-\epsilon) -p\delta \rceil\leq \lceil (q+2\epsilon) \delta \rceil$.
Thus, $(S_1,S_2)$ is $\langle \lfloor{(p+2\epsilon)\delta}\rfloor,\lceil{(q+2\epsilon)\delta}\rceil \rangle$ 
dominating at time $t_0$. Observe that as edges are added to $G$, the depth of vertices in $\outbfs(S_1)$ 
and $\inbfs(S_2)$ can only decrease with time, so the set-pair $(S_1,S_2)$ must also be 
$\langle \lfloor{(p+2\epsilon)\delta}\rfloor,\lceil{(q+2\epsilon)\delta}\rceil \rangle$ dominating at time $t$.
\end{proof}

We now present algorithms that
for a given $\epsilon\in[0,1/2]$, and integers $n_p,n_q\geq 1$ satisfying $n_pn_q=8n\log n$,
incrementally (and decrementally) maintains for an $n$-vertex graph $G$
a triplet $(S_1,S_2,a)\in \P(V)\times \P(V)\times V$ such that at any time instance $t$,
\begin{enumerate}[{\normalfont (i)}]
\item $|S_1|=n_p$, and $S_2=\nin_{t_0}(a,n_q)$ for some $t_0\leq t$,
\item $\depth_t(\outbfs(S_1))$, $\depth_t(a,\outbfs(S_1))$, $\depth_{t_0}(a,\outbfs(S_1))$
		lies in range $[\ell_0(1-\epsilon),\ell_0(1+\epsilon)]$, where $\ell_0=\depth_{t_0}(\outbfs(S_1))$, and
\item $\depth_{t}(\inbfs(a))\geq (1-\epsilon)\depth_{t_0}(\inbfs(a))$.
\end{enumerate}

\noindent{\em Incremental scenario.}
We first discuss the  incremental scenario. The main obstacle in this setting is 
to dynamically maintain a vertex $a$ having large depth in $\outbfs(S_1)$.
We initialize $S_1$ to a uniformly random subset of $V$ containing $n_p$ vertices, and
store in $\ell_0$ the depth of tree $\outbfs(S_1)$. Next we compute a set $\textsc{far}$ 
that consist of all those vertices whose distance from $S_1$ is at least $(1-\epsilon)\ell_0$,
and set $A$ to be a uniformly random subset of $\textsc{far}$ of size $\min\{8\log n,|\textsc{far}|\}$.
We initialize $a$ to any arbitrary vertex in set $A$, and set $S_2$ to $\nin(a,n_q)$.
Throughout the algorithm whenever $S_1\cap S_2$ is empty, then we 
recompute $S_1,\ell_0,A,a$, and $S_2$. The probability of such an event is inverse polynomial in $n$.

We use Theorem~\ref{theorem:ES} to dynamically maintain $\depth(\outbfs(S_1))$
and depth of individual vertices in $\outbfs(S_1)$. This takes $O(mD)$ time for any fixed $S_1$.
Whenever $\depth(\outbfs(S_1))$ falls below the value $(1-\epsilon)\ell_0$, then we recompute 
$\ell_0, A,a$ and $S_2$. For any fixed $S_1$, this happens at most $O(\epsilon^{-1} \log n)$ times, and
so takes $O(m\epsilon^{-1} \log n)$ time  in total.
Whenever depth of a vertex lying in $A$ falls below the value $(1-\epsilon)\ell_0$, then we remove
that vertex from $A$. This step takes $O(m|A|)=O(m\log n)$ time in total.
If $\depth(a,\outbfs(S_1))$ falls below the value $(1-\epsilon)\ell_0$, then we replace $a$
by an arbitrary vertex in $A$, and recompute $S_2$.

If $A$ becomes empty and $\depth(\outbfs(S_1))$ is still greater than $(1-\epsilon)\ell_0$,
then we recompute $\textsc{far},A,a$, and $S_2$. Observe that for any fixed $\ell_0$
this happens at most $\log n$ times. This is because if $\textsc{far}_1$ and $\textsc{far}_2$
is a partition of $\textsc{far}$ such that the depth of all vertices  in $\textsc{far}_1$
falls below $(1-\epsilon)\ell_0$ earlier than the vertices in  $\textsc{far}_2$, then with high probability
$A$ has a non-empty intersection with $\textsc{far}_2$. This holds true as we assume adversarial model
in which edge insertions are independent of choice of $A$.  Thus with high probability. each time $A$ is
recomputed the size of set $\textsc{far}$ decreases by at least half, assuming $\ell_0$ remains fixed.
Since $\ell_0$ changes at most $\epsilon^{-1} \log n$ times, the set $A$ is recomputed at most
$\epsilon^{-1} \log^2 n$ times, and vertex $a$ can thus change $O(\epsilon^{-1} \log^2 n |A|)=
O(\epsilon^{-1} \log^3 n)$ times. 

Finally, for vertex $a$ we maintain $\depth(\inbfs(a))$ using ES-tree. 
Since, $a$ changes at most $O(\epsilon^{-1} \log^3 n)$ times, total time for maintaining $\depth(\inbfs(a))$ 
is $O(mD\epsilon^{-1} \log^3 n)$. Whenever $\depth(\inbfs(a))$
falls by a factor of $(1-\epsilon)$, then we re-set $S_2$ to $\nin(a,n_q)$. For a fixed $a$, $S_2=\nin(a,n_q)$
is updated at most $O(\epsilon^{-1} \log n)$ times. So in total $S_2$ changes at most $O(\epsilon^{-2} \log^4 n)$
times, and the total time for maintaining set $S_2$, throughout the edge insertions is $O(m\epsilon^{-2} \log^4 n)$.
Thus, the total time taken by the algorithm is $O(\epsilon^{-1}D \log^3 n+\epsilon^{-2} \log^4 n)$, where $D$ denotes
the maximum diameter of $G$ throughout the sequence of edge updates. Also the expected number of times the
triplet $(S_1,S_2,a)$ changes is $O(\epsilon^{-2} \log^4 n)$.\\

\noindent{\em Decremental Scenario.}
We now discuss the simpler scenario of edge deletions.
As before, we initialize $S_1$ to be a uniformly random subset of $V$ containing $n_p$ vertices.
Next we compute $\outbfs(S_1)$ and set $a$ to be an arbitrary vertex having maximum depth in $\outbfs(S_1)$.
Also $S_2$ is set to $\nin(a,n_q)$. We store in $\ell_0$ the depth of tree $\outbfs(S_1)$, and as
in incremental setting use Theorem~\ref{theorem:ES} to dynamically maintain the depth of $\outbfs(S_1)$.
This takes $O(mD)$ time in total.
Whenever $\depth(\outbfs(S_1))$ exceeds the value $(1+\epsilon)\ell_0$, then we recompute 
$\ell_0, a$ and $S_2$. For any fixed $S_1$ such an event happens at most $O(\epsilon^{-1} \log n)$ times, and
takes in total $O(m\epsilon^{-1} \log n)$ time.
Also whenever $S_1\cap S_2$ is non-empty, then we 
recompute $S_1,S_2,a,\ell_0$, and reinitialize the Even and Shiloach data-structure.
The probability of such an event is inverse polynomial in $n$.
So the expected amortized update time for edge deletions is $O(D+\epsilon^{-1} \log n)$.
Also if $t_0$ is the time when $\ell_0,a,S_2$ were last updated and $t$ is the current time then
$\depth_{t_0}(a,\outbfs(S_1))=\ell_0$,
$\depth_{t}(\outbfs(S_1)),\depth_{t}(a,\outbfs(S_1))\in [\ell_0,(1+\epsilon)\ell_0]$. Thus the conditions (i) and (ii) hold.
Also $\depth(\inbfs(a))$ only increases with time, so condition (iii) trivially holds.
So, the amortized update time of the procedure is $O(D+\epsilon^{-1} \log n)$ in the decremental scenario, 
where $D$ denotes the maximum diameter of G throughout the sequence of edge updates;
and the expected number of times when triplet $(S_1,S_2,a)$ changes $O(\epsilon^{-2} \log^4 n)$. 

The following theorem is immediate from the above discussion and Lemma~\ref{lemma:dynamic-dom-set-1}.

\begin{theorem}
For any $\epsilon\in[0,1/2]$, and any integers $n_p,n_q\geq 1$ satisfying $n_pn_q=8n\log n$,
there exists an algorithm that incrementally/decrementally maintains for an $n$-vertex directed graph
a set-pair $(S_1,S_2)$ of size-bound $\langle n_p,n_q \rangle$ which is 
$\langle \lfloor{(p+2\epsilon)~\inecc(a)}\rfloor,\lceil{(q+2\epsilon)~\inecc(a)}\rceil \rangle$-%
dominating, for some $a\in V$, and any arbitrary fractions $p,q>0$ satisfying $p+q=1$.
 
The expected amortized update time of the algorithm is $O(\epsilon^{-1}\maxD \log^3 n+\epsilon^{-2} \log^4 n)$ 
in incremental setting and $O(\maxD+\epsilon^{-1} \log n)$ in decremental setting,
where, $\maxD$ denotes the maximum diameter of the graph throughout the sequence of edge updates.
Also, the algorithm ensures that with high probability the triplet $(S_1,S_2,a)$ changes 
at most $O(\epsilon^{-2} \log^4 n)$ times in the incremental setting, and 
at most $O(\epsilon^{-1} \log n)$ times in the decremental setting.
\label{theorem:dynamic-dom}
\end{theorem}


\subsection{Dynamic Algorithms for $1.5$-Diameter Spanners}

We consider two model for maintaining the diameter spanners, namely, the {\em explicit} model and the {\em implicit} model.
The {\em explicit} model maintains at each stage all the edges of a diameter spanner of the current graph.
In the model of {\em implicitly} maintaining the diameter spanner, the goal is to have a 
data-structure that efficiently supports the following operations:
(i) $\textsc{update}(e)$ that adds to or remove from the graph $G$ the edge $e$, and
(ii) $\textsc{query}(e)$ that checks if the diameter spanner contains edge $e$.

We first consider the explicit maintenance of diameter spanners.

Let $\A$ be an algorithm that uses Theorem~\ref{theorem:dynamic-dom} to incrementally (or decrementally)
maintain at any time $t$, a 
$\langle\lfloor(1/2+\epsilon)\inecc_t(a)\rfloor,\lceil (1/2+\epsilon)\inecc_t(a)\rceil\rangle$-dominating 
set-pair $(S_1,S_2)$ of size bound $\langle \sqrt{n\log n},\sqrt{n\log n} \rangle$, where $a\in V$.
We dynamically maintain a subgraph $H$ which is union of $\inbfs(s)$ and $\outbfs(s)$, for $s\in S_1\cup S_2$.
This takes in total $O(m\maxD|S_1\cup S_2|)=O(m\cdot \maxD\sqrt{n\log n})$ time,
where $\maxD$ is the maximum diameter of the graph throughout the sequence of edge updates. 
Observe that similar to Theorem~\ref{theorem:static_diam1}, it can be shown that at any time instance
subgraph $H$ is a $(1/2+\epsilon)$-diameter spanner of $G$, and it contains at most $O(n\sqrt{n\log n})$ edges.
Let $\emph{T}$ be the expected amortized update time of $\A $
for maintaining $(S_1,S_2)$, and let $\emph{C}$ be the total number of times the  
pair $(S_1,S_2)$ changes throughout the algorithm run.
Then the total time for maintaining $H$ is 
$O(\emph{C}\cdot m\cdot \maxD\sqrt{n\log n}+m\cdot \emph{T})$.
On substituting the values of $\emph{C}$ and $\emph{T}$
from Theorem~\ref{theorem:dynamic-dom},
we get that the expected amortized update time of $\A $ is
$O(\epsilon^{-1}\maxD\sqrt{n}\log^{1.5} n)$
for the decremental setting, and $O(\epsilon^{-2}\maxD\sqrt{n}\log^{4.5} n)$
for the incremental setting. 

For the scenario when $\maxD$ is large we alter our algorithm as follows.
Let $\thD$ be some threshold value for diameter. 
We maintain a 2-approximation of $\dia(G)$, say $\delta$, by
dynamically maintaining for an arbitrarily chosen vertex $z$, the value $\depth(\inbfs(z))+\depth(\outbfs(z))$.
This by Theorem~\ref{theorem:ES} takes $O(mn)$ time in total. 
We now explain another algorithm $\cal B$ which will be effective when $\delta\geq 4\thD$. 
We sample a uniformly random subset $W$ of $V$ containing $(8n\log n/\thD)$ 
vertices, and maintain at each stage a subgraph $H_{\B}$
which is union of $\inbfs(w)$ and $\outbfs(w)$, for $w\in W$.
Also we maintain the value $\depth(\outbfs(W))$.
If $\delta\geq 4\thD$, but $\depth(\outbfs(W))\nleq \thD$, we re-sample $W$.
When $\delta\geq 4\thD$, then with high probability at each time instance, set $W$ intersects $\pi(x,y)$ for every 
$x,y\in V$ that satisfy $d_G(x,y)\geq \thD$, and thus $\depth(\outbfs(W))\leq \thD$.
This shows that the expected number of re-samplings for $W$ is $O(1)$, and
the total expected runtime of $\B$ is $O(mn|W|)=O(mn^2\log n/\thD)$.
Since $\depth(\outbfs(W))\leq \thD\leq \delta/4\leq \dia(G)/2$, it follows that in this case
the distance between any two vertices in 
$H_{\B}$ is at most $1.5\dia(G)$.
As long as $\delta\leq 4\thD$, we use algorithm $\A$ to maintain a $(1.5+\epsilon)$-diameter spanner, 
we denote the corresponding subgraph by notation $H_{\A}$.
Thus $\A$ takes in total $O(\epsilon^{-2}\thD m\sqrt{n}\log^{4.5} n)$ time for incremental setting, and 
$O(\epsilon^{-1}\thD m\sqrt{n}\log^{1.5} n)$ time for decremental setting. 
On optimizing over $\thD$, we get that the amortized update time of the combined algorithm is
$O(\epsilon^{-1} n^{1.25}\log^{2.75} n)$ for incremental setting, and 
$O(\epsilon^{-0.5} n^{1.25}\log^{1.25} n)$ for decremental setting.
Thus, we obtain the following theorem.

\begin{theorem}
For any $\epsilon\in[0,1/2]$ and any  incrementally/decrementally changing graph on $n$ vertices, 
there exists an algorithm for maintaining a $(1.5+\epsilon)$-diameter spanner containing at most $O(n^{3/2}\sqrt{\log n})$ edges.
 
The expected amortized update time of the algorithm is $O((1/\epsilon^{2}) \sqrt{n}\maxD\log^{4.5} n)$
for incremental setting and $O((1/\epsilon) \sqrt{n}\maxD\log^{1.5} n)$ for decremental setting,	
where, $\maxD$ denotes the maximum diameter of the graph throughout the run of the algorithm.
Moreover, when $\maxD$ is large, the algorithm can be altered so that the expected amortized update time 
is $O(\epsilon^{-1} n^{1.25}\log^{2.75} n)$ for the incremental setting, and $O(\epsilon^{-0.5} n^{1.25}\log^{1.25} n)$ 
for the decremental setting.
\label{theorem:dynamic-spanner-1}
\end{theorem}

We now present the algorithm for implicitly maintaining diameter spanner.
Let $\A$ be a Monte-Carlo variant of Theorem~\ref{theorem:dynamic-dom} to incrementally/decrementally
maintain a $\langle\lfloor(1/2+\epsilon)D\rfloor,\lceil (1/2+\epsilon)D\rceil\rangle$-dominating set-pair 
$(S_1,S_2)$ of size bound $\langle \sqrt{n\log n},\sqrt{n\log n} \rangle$. So $\A$ 
takes in total $O(\epsilon^{-1}mn\log^3 n+m\epsilon^{-2}\log^4 n)$ time for incremental setting, and 
$O(mn+m\epsilon^{-1} \log n)$ time for decremental setting. 
We also maintain a data-structure for dynamic all-pairs shortest-path problem. 
For edge-insetions only case, Ausiello et al.~\cite{AusielloIMN:91} gave an $O(n^3 \log n)$ time algorithm
that answers any distance query in constant time, and for edge-deletions only case, Baswana et al.~\cite{BaswanaHS:07} 
gave an $O(n^3 \log^2 n)$ time Monte-Carlo algorithm that again answers any distance query in constant time.
Now in order to check whether or not an edge $e=(u,v)$ lies in $H$, it suffices to check
whether or not $e$ is present in either $\inbfs(s)$ or $\outbfs(s)$, for some $s\in S$.
We can assume that edge weights are slightly perturbed so that no two distances are identical in $G$.
Therefore $e=(u,v)$ lies in $\outbfs(s)$, for some $s$, if and only if $d_G(s,v)=d_G(s,u)+d_G(u,v)$.
Since the distances queries can be answered in $O(1)$ time, in order to check whether or not $e$ lies in $H$,	
we perform in the worst case $O(|S_1\cup S_2|)=O(\sqrt{n\log n})$ distance queries.
From the above, we obtain the following theorem.
	
\begin{theorem}
There exists a data-structure that for any incrementally/decrementally changing 
$n$-vertex directed graph and any $\epsilon\in \big[\frac{\log n}{n}, \frac{1}{2}\big]$, implicitly 
maintains a $(1.5+\epsilon)$-diameter spanner containing at most $O(n^{3/2}\sqrt{\log n})$ edges.
The total time taken by $\textsc{update}$ operations is $O(\epsilon^{-1} n^3 \log^3(n))$ for incremental
setting, and $O(n^3 \log^2(n))$ for decremental setting. Each $\textsc{query}$
operation takes $O(\sqrt{n\log n})$ time in the worst case, and the answers 
are correct with high probability (i.e., failure probability is inverse polynomial in $n$).
\label{theorem:dynamic-spanner-2}
\end{theorem}

\subparagraph*{Implication} 

As a byproduct of our dynamic $(1.5+\epsilon)$-diameter spanner algorithm, we also obtain a dynamic algorithm for maintaining
a $(1.5+\epsilon)$-approximation of the graph's diameter. 
As mentioned in Section~\ref{sec:dynamic-spanners}, the latter was very recently studied also by Ancona {\em et al.}~\cite{arxiv:diam}, obtaining related (slightly different) bounds (see Section~\ref{sec:dynamic-spanners} and~\cite{arxiv:diam} for more details). 
Since our algorithm is implied from our dynamic diameter spanner techniques, we feel it worth describing the derivation below.

Let $\cal A$ be an algorithm that uses Theorem~\ref{theorem:dynamic-dom} to dynamically 
maintain triplet $(S_1,S_2,a)$ such that at any time instance set-pair $(S_1,S_2)$ 
is $\langle \lfloor{(1/2+\epsilon)~\inecc(a)}\rfloor,\lceil{(1/2+\epsilon)~\inecc(a)}\rceil \rangle$%
-dominating and has size bound $\langle \sqrt{n\log n},\sqrt{n\log n}\rangle$. 
Let $\emph{T}({\cal A})$ be the expected amortized update time of $\cal A$
for maintaining $(S_1,S_2,a)$. Also let $\emph{C}({\cal A})$ be the total number of times the triplet 
$(S_1,S_2,a)$ changes throughout the run of the algorithm.

Since $(S_1,S_2)$ is $\langle \lfloor{(1/2+\epsilon)~\dia(G)}\rfloor,\lceil{(1/2+\epsilon)~\dia(G)}\rceil \rangle$-%
dominating, for any pair of vertices $x,y$ in $V$, we have $d_G(x,y)\leq \max_{s\in S_1\cup S_2}
(1.5+\epsilon)\max\{\inecc(s),\outecc(s)\}$, which in turn is bounded by
$\lceil (1.5+\epsilon)\dia(G)\rceil$. Thus, to dynamically maintain a 1.5-approximation of 
diameter it suffices to maintain $\depth(\inbfs(s))$ and $\depth(\outbfs(s))$ for each $s\in S_1\cup S_2$.
This by Theorem~\ref{theorem:ES} takes $O(m\maxD)$ time in total for any $s\in S_1\cup S_2$,
where, $\maxD$ denotes the maximum diameter of graph throughout the sequence of edge updates.
Observe that the pair $(S_1,S_2)$ also alters at most $\emph{C}({\cal A})$ times.
So the total time for maintaining a 1.5-approximation of diameter is 
$O(|\emph{C}({\cal A})|m\maxD\sqrt{n\log n})+m\emph{T}({\cal A}))$.
On substituting the values of $\emph{C}({\cal A})$ and $\emph{T}({\cal A})$
from Theorem~\ref{theorem:dynamic-dom},
we get that the expected amortized update time is
$O(\epsilon^{-1}\maxD\sqrt{n}\log^{1.5} n)$
for the decremental setting, and $O(\epsilon^{-2}\maxD\sqrt{n}\log^{4.5} n)$
for the incremental setting. 

%
 
\section{Additional Sparse Diameter Spanners Constructions}
\label{section:upper-bounds}

In this section, we present additional constructions of diameter spanners, with different size-stretch trade-offs. 

\subsection{$5/3$-Diameter Spanner}
We first present construction of $5/3$-diameter spanners that are sparser than the $1.5$-%
diameter spanners whenever $D=o(\sqrt{n})$.

\begin{theorem}
For any directed graph $G=(V,E)$ with diameter $D$, in 
$\widetilde O(mn^{1/3}(D+n/D)^{1/3})$ expected time\footnote{Though the computation time  
is a function of $D$, 	the algorithm does not 
need to apriori know the value $D$.}  we can compute a subgraph $H=(V, E'\subseteq E)$
satisfying $\dia(H)\leq \lceil 5D/3\rceil$ that contains at most $O(n^{4/3}(\log n)^{2/3}D^{1/3})$ edges,
where $n$ and $m$ respectively denotes the number of vertices and edges in $G$.
\footnote{As in Theorem~\ref{theorem:static_diam1}, all our spanner constructions work also 
for edge-weighted graphs with non-negative weights,
by replacing every use of BFS with Dijkstra's algorithm.
The runtime of our construction is increased by at most a $\log n $ factor, and the stretch factor
of the spanner $H$ only suffers an additive $W$ term, where
$W$ is the maximum edge weight in the graph.}
\label{theorem:static_diam2}
\end{theorem}

\begin{proof}
Let $\alpha$ be a parameter to be chosen later on. The construction of $H$ is presented in 
Algorithm~\ref{Algorithm:diam-5/3}. Consider any two vertices $x,y\in V$.
If $A_1$ is $\lceil 2D/3\rceil$-out-dominating set, then 
$d_G(s,y)\leq  \lceil 2D/3\rceil$ for some $s\in A_1$. Thus 
$d_H(x,y)\leq d_H(x,s)+d_H(s,y)=d_G(x,s)+d_G(s,y)\leq D + \lceil 2D/3\rceil =  \lceil 5D/3\rceil$.
Similarly, if $B_2$ is $\lceil 2D/3\rceil$-in-dominating set, then it can be shown that $d_H(x,y)$
is bounded by $\lceil 5D/3\rceil $.
Let us next suppose that neither $A_1$ is $\lceil 2D/3\rceil$-out-dominating, nor $B_2$ is 
$\lceil 2D/3\rceil$-in-dominating. Then $A_2$ is $D/3$-in-dominating and $B_1$ is 
$D/3$-out-dominating. So  $d_G(x,A_2),d_G(B_1,y)\leq D/3$. Since $H$ contains 
$\inbfs(A_2)$ and $\outbfs(B_1)$, there must exists $s_x\in A_2$ and $s_y\in B_1$ 
such that $d_H(x,s_x)=d_G(x,s_x)=d_G(x,A_2)$ and $d_H(s_y,y)=d_G(s_y,y)=d_G(B_1,y)$. 
Since $H$ contains the shortest path between each pair of vertices in $A_2\times B_1$, we 
obtain that $d_H(s_x,s_y)=d_G(s_x,s_y)\leq D$. Therefore, 
$d_H(x,y)\leq d_H(x,s_x)+d_H(s_x,s_y)+d_H(s_y,y) = d_G(x,s_x)+d_G(s_x,s_y)+d_G(s_y,y) \leq 5D/3$.
\begin{algorithm}[!ht]
$H\gets $ an empty graph\;
$(A_1,A_2)\gets$ $\langle \lceil 2D/3\rceil,D/3\rangle$-dominating-set-pair of size-bound $\langle \alpha\log n, n/\alpha \rangle$\;
$(B_1,B_2)\gets$ $\langle D/3,\lceil 2D/3\rceil \rangle$-dominating-set-pair of size-bound $\langle n/\alpha,\alpha\log n \rangle$\;
Add to $H$ the trees $\inbfs(A_2)$ and $\outbfs(B_1)$\;
\lForEach {$s\in A_1\cup B_2$} {add to $H$ union of $\inbfs(s)$ and $\outbfs(s)$}
\lForEach {$(u,v)\in A_2\times B_1$} {add the edges of the shortest path $\pi_G(u,v)$ to $H$}
\Return $H$\;
\caption{5/3-Diameter Spanner Construction}
\label{Algorithm:diam-5/3}
\end{algorithm}
 Let us first analyse size of $H$.
We have $O(\alpha \log n)$ shortest-path trees that  require a total of $O(n \alpha \log n)$ edges.
The shortest paths between all pairs in $A_2\times B_1$ use in total $O(n^2D/\alpha^2)$ edges.
Thus, the total number of edges in $H$ is $O(n \alpha \log n + n^2D/\alpha^2)$. This is minimized when
$\alpha = \Theta\big((nD/\log n)^{1/3}\big)$. Therefore, the total number of edges in $H$ is 
$O(n^{4/3}D^{1/3}\log^{2/3} n)$.
Observe that in order to compute $\alpha$, it suffices to have an estimate of $D$.
We can easily compute a $2$-approximation for the diameter $D$ in $O(m)$ time,
since for any arbitrary vertex $w\in V$,  
$D\leq \depth(\inbfs(w))+\depth(\outbfs(w))\leq 2D$, and the depth of $\inbfs(w)$ 
and $\outbfs(w)$ are computable in $O(m)$ time. 
We now analyse the running time of each step in Algorithm \ref{Algorithm:diam-5/3}. 
Steps 2 and 3: By Lemma~\ref{lemma:dom-set-rand}, the time to compute the set-pairs 
$(A_1,A_2)$ and $(B_1,B_2)$ is $O(m)$ on expectation. Step 4: This step just takes 
$O(m)$ time. Step 5 and 6: For each vertex $s \in A_1 \cup B_2 \cup A_2\cup B_1$, 
the BFS trees $\inbfs(s)$ and $\outbfs(s)$ can be computed in $O(m)$ time. So, 
this step can be performed in $O(m\cdot |A_1 \cup B_2 \cup A_2\cup B_1|)$ time. 
Overall, the total expected runtime of the  algorithm is $O(m(|A_1\cup A_2\cup B_1\cup B_2|))
=O(m(\alpha\log n+n/\alpha))=\widetilde O(mn^{1/3}D^{1/3}+mn^{2/3}/D^{1/3})
=\widetilde O(mn^{1/3}(D^{1/3}+(n/D)^{1/3}))=\widetilde O(mn^{1/3}(D+n/D)^{1/3})$.
\end{proof}

\subparagraph{Dynamic maintenance of 5/3-diameter spanner.}
In order to dynamically maintain a $5/3$-diameter spanner we first state a
 lemma which is a generalization of Theorem~\ref{theorem:static_diam2}.

\begin{lemma}
Let $G$ be an $n$-vertex directed graph with diameter $D$, and $(A_1,A_2)$ and $(B_1,B_2)$ be respectively 
$\langle \lceil (2/3+\epsilon)D\rceil,(1/3+\epsilon)D\rangle$ and
$\langle (1/3+\epsilon)D,\lceil (2/3+\epsilon)D\rceil\rangle$ dominating-set-pairs  
of size bounds $\langle \alpha\log n,n/\alpha \rangle$ and $\langle n/\alpha,\alpha \log n \rangle$,
where $\alpha=(nD)^{1/3}$.
Also let $H$ be a subgraph of $G$ consisting of 
\begin{itemize} 
\item $\inbfs(s)$ and $\outbfs(s)$, for $s\in A_1\cup B_2$,
\item the shortest paths $\pi_G(u,v)$, for each $(u,v)\in A_2\times B_1$, 
\item $\inbfs(A_2)$ and $\outbfs(B_1)$. 
\end{itemize}
Then $\dia(H)\leq \lceil (5/3+\epsilon)D\rceil $ and $H$ has
at most $O(n^{4/3}D^{1/3} \log n)$ edges.
\label{lemma:static2_gen}
\end{lemma}

The following theorem presents our bounds for dynamic maintenance a $5/3$-diameter spanner
(we omit the proof as it is analogous to that of Theorem~\ref{theorem:dynamic-spanner-1}).

\begin{theorem}
For any $\epsilon\in[0,1/2]$, there exists an algorithm for incrementally/decrementally maintaining
a $(5/3+\epsilon)$-diameter spanner of an $n$-vertex directed graph with at most $O(n^{4/3}\maxD^{1/3}\log n)$ edges,
where, $\maxD$ denotes the maximum diameter of the graph throughout the run of the algorithm.
The expected amortized update time of the algorithm is $O(\epsilon^{-2} n^{1/3}\maxD(\maxD+n/\maxD)^{1/3}\log^5 n)$
for the incremental setting and $O(\epsilon^{-1} n^{1/3}\maxD(\maxD+n/\maxD)^{1/3}\log^2 n)$ for the decremental setting.
\label{theorem:dynamic-spanner-5by3-stretch}
\end{theorem}

\subsection{Diameter Spanner with $\O(n)$ Edges}
In Theorem~\ref{theorem:3/2-diam-up-lb}(b) and Theorem~\ref{theorem:5/3-diam-up-lb}(b),
we show a lower bound on number of edges in diameter spanners of stretch respectively $3/2$ and $5/3$,
for graphs with low diameter. A natural question to ask is if for graphs with large diameter it is possible to
obtain low-stretch diameter spanners that are also sparse. In this subsection, we 
positively address this problem.
 
We first show construction of an $\widetilde O(n)$ size spanner with additive stretch.

\begin{lemma}
For any $d>0$, and any $n$-vertex directed graph $G=(V,E)$, we can compute a subgraph $H=(V,E'\subseteq E)$ 
with $O(n+dn^\frac{2}{3}\cdot\min\{n^\frac{1}{3}\log n,d\log^2 n\})$ edges satisfying $\dia(H)\leq \dia(G)+n/d$.
\label{theorem:additive_1}
\end{lemma}

\begin{proof}
Let $S$ be random set of $8d\log n$ vertices.
We first check that $S$ has non-empty intersection with $\nout(w,n/2d)$ and $\nin(w,n/2d)$,
for each $w\in V$, if not then re-sample $S$. The expected computation time of $S$ is $O(m)$.
Next we initialize $H$ to union of trees $\inbfs(S)$ and $\outbfs(S)$. 

If $n^\frac{1}{3}\log n\leq d\log^2 n$, then we add to $H$ tree $\outbfs(s)$, for each $s\in S$. 
So the number of edges in $H$ is $O(n|S|)=O(nd\log n)$.
To prove correctness consider any two vertices $x,y\in V$.
There must exists $s\in S$ such that, $d_H(x,s)=d_H(x,S) = d_G(x,S)\leq n/d$
(the last inequality holds since $S$ intersects $\nout(x,n/d)$ and $G$ is unweighted).
Also $d_H(s,y)\leq \dia(G)$. Thus, $d_H(x,y)\leq d_H(x,s)+d_H(s,y)\leq \dia(G)+n/d$.

Let us next consider the case $d\log^2 n\leq n^\frac{1}{3}\log n$. In this case
 we add to $H$ a pair-wise distance preserver for 
each pair of nodes in $S\times S$. Bodwin~\cite{Bodwin17} showed that for any set $S(\subseteq V)$ 
in a directed graph, we can compute a sparse subgraph with at most $O(n+n^{2/3}|S\times S|)$ edges 
that preserves distance between each node pair in $S\times S$. So the total number of edges in 
subgraph $H$ is $O(n+n^{2/3}|S\times S|)=O(n+n^{2/3}d^2\log^2n)$. Now to prove the 
correctness consider any two vertices $u,v\in V$, let $x_u,x_v\in S$ be such that 
$x_u\in \nout(u,n/2d)$ and $x_v\in \nin(v,n/2d)$. Then $d_H(u,s_u)\leq n/2d$, 
$d_H(s_v,v)\leq n/2d$, and $d_H(s_u,s_v)\leq \dia(G)$. Thus, $d_H(u,v)\leq \dia(G)+n/d$.
\end{proof}

On substituting $d=n^{1/6}$ in theabove lemma, we obtain a
construction of diameter spanner with $(1+o(1))$ stretch
that contains $\O(n)$ edges, whenever $D=\omega(n^{5/6})$.



\begin{theorem}
For any $n$-vertex directed graph $G=(V,E)$ satisfying $\dia(G)=\omega(n^{5/6})$, we can compute a 
$(1+o(1))$-diameter spanner of $G$ containing at most $O(n\log^2n)$ edges.
\end{theorem}

\subsection{General (low-stretch or low-size)-Diameter Spanners}
We show that for any directed graph $G$ we can either 
(i) compute a diameter spanner with arbitrarily low stretch, or
(ii) compute a diameter spanner with arbitrarily low size.

\begin{theorem}
Let $n_p,n_r>1$ be integers satisfying $n_pn_r= 8n\log n$, and $p,r>0$ be fractions satisfying $p+r=1$. 
For any directed graph $G=(V,E)$ with $n$ vertices and $m$ edges, in $O(m\max\{n_p,n_r\})$ 
expected time, we can compute a subgraph $H=(V,E_0\subseteq E)$ satisfying at least one of the following conditions
:\\[1mm]
(i) $|E_0|=O(nn_p)$ and $\dia(H)\leq \lceil (1+p) ~\dia(G)\rceil $.\\[1mm]
(ii) $|E_0|=O(nn_q)$ and $\dia(H)\leq \lceil (1+r) ~\dia(G)\rceil $.
\label{theorem:pq_diam}
\end{theorem}
\begin{proof}
Let $D$ denote the diameter of $G$. Let $(S_1,S_2)$ 
be a $\langle \lfloor p~\inecc_G(a)\rfloor,\lceil r~\inecc_G(a)\rceil\rangle$-dominating-set-pair of 
size bound $\langle n_p,n_r\rangle$ obtained from Lemma~\ref{lemma:dom-set-rand} for 
some $a\in V$ and some integers $n_p,n_r>1$ satisfying $n_pn_r= 8n\log n$.
Let $H_1$ (respectively $H_2$) be the union of the trees $\inbfs(s)$ and $\outbfs(s)$, 
for each $s\in S_1$ (respectively $S_2$). The time for computing $H_1$ and $H_2$ 
is derived from $|S_1\cup S_2|$ BFS computations, plus the time for finding the dominating 
set-pair $(S_1,S_2)$, which in total is $O(m(n_p+n_r))$ on expectation. Note that the
graph $H_1$ contains $O(nn_p)$ edges and $H_2$ contains $O(nn_r)$ edges. 
Now consider any two vertices $x,y\in V$.
If $S_1$ is $\lfloor p~\inecc_G(a)\rfloor $-out-dominating, then there exists $s\in S_1$ such that 
$d_G(s,y)\leq  p~\inecc_G(a)\leq pD$. Since $d_{H_1}(x,s)=d_G(x,s)\leq  D $, and 
$d_{H_1}(s,y)=d_G(s,y)\leq pD$,  we have $d_{H_1}(x,y)\leq (1+p)D$.
Similarly, if $S_2$ is $\lceil r~\inecc_G(a)\rceil$-in-dominating, then $d_{H_2}(x,y)\leq \lceil (1+r)D\rceil$.
Thus the claim follows.
\end{proof}

As a corollary we obtain the following result.

\begin{corollary}
For any arbitrarily small fractions $\epsilon,\delta>0$, and any given directed graph $G$, 
in~$\widetilde O(mn^{1-\epsilon})$ expected time, at least one of the following subgraphs
can be computed. \\[1mm]
(i) a $(1+\delta)$-diameter spanner of $G$ containing at most $O(n^{2-\epsilon}\sqrt{\log n})$ edges.\\[1mm]
(ii) a $(2-\delta)$-diameter spanner of $G$ containing at most $O(n^{1+\epsilon}\sqrt{\log n})$ edges.
\end{corollary}

\section{Eccentricity Spanners and Radius Spanners}

\subsection{An $\O(n)$-Size 2-Eccentricity Spanner Computation in $\O(m)$ Time}
In order to obtain a $2$-eccentricity spanner in near optimal time\footnote{
Observe that a graph $H$ which is union of $\inbfs(c)$ and $\outbfs(c)$,
where $c$ is the centre~(a vertex of minimum out-eccentricity) of $G$, is a $2$-eccentricity spanner
containing just $2n$ edges. 
However, the computation time of $H$ is large since the best known
algorithm for computing a centre of directed weighted graphs takes $O(mn)$ time. 
}, we first show that for any $n$-vertex graph $G$,
we can compute in $\O(m)$ time a set $S$ containing $O(\log^2 n)$ vertices 
such that $\depth(\outbfs(S))$ is at most $\rad(G)$. 
(Our results hold for the general setting of directed weighted graphs).

\begin{algorithm}[!ht]
$B_k\gets V$\;
\For{$i=k-1$ to $1$}{
$A_i\gets $ uniformly random subset of $B_{i+1}$ of size $\min\{8n^{1/k}\log n,|B_{i+1}|\}$\;
$a_i\gets $ vertex of maximum depth in $\outbfs(A_i)$\;
$B_i\gets $ a subset of $B_{i+1}$ containing the $n^{(i/k)}$-closest incoming vertices to $a_i$ that lie in $B_{i+1}$\;
\lIf{$A_i\cap B_i=\emptyset$}{ go to step 3 to re-sample $A_i$, and next recompute $a_i$ and $B_i$}
}
$S\gets B_1\cup A_1\cup A_2\cup\cdots \cup A_{k-1}$\;
$H\gets $ union of $\inbfs(s)$ and $\outbfs(s)$, for $s\in S$\;
\Return $H$\;
\caption{2-Eccentricity Spanner Construction}
\label{Algorithm:ecc-2}
\end{algorithm}

Let $k$ be a parameter to be chosen later on. The construction of set $S$ is very simple and 
presented in Algorithm~\ref{Algorithm:ecc-2}. The expected runtime of the algorithm is $O(mk+m|S|)$. 
To observe this note that $|B_{i+1}|=n^{(i+1)/k}$.
Now we take $A_i$ to be uniformly random subset of $B_{i+1}$ of size at most $8n^{1/k} \log n$,
and $B_i$ contains those $n^{i/k}$ closest incoming vertices to $a_i$ that lie in $B_{i+1}$.
Since $B_i$, $A_i$ are both subsets of $B_{i+1}$, the expected number of re-samplings in step 6 for 
each $i\in[1,k-1]$ is $O(1)$. So the total time taken by steps 2-6 is $O(mk)$ on expectation, 
and the time taken by steps 7-9 is $O(m|S|)$. 

We next prove the correctness of the algorithm through the following lemmas.

\begin{lemma} For any index $i\in[1,k]$, the set-pair $(A_i,B_i)\in \P(B_{i+1})\times \P(B_{i+1})$ 
is of size bound $\langle n^{1/k}\log n,n^{i/k}\rangle$
and satisfy that for each $x\in B_{i+1}$, either $\depth(\outbfs(A_i))\leq \outecc(x)$, or $x\in B_i$.
\label{lemma:ecc-new-1}
\end{lemma}
\begin{proof}
Suppose $d$ is an integer satisfying $\depth(\outbfs(A_i))=\depth(a_i,\outbfs(A_i))> d$. 
Then $\inbfs(a_i,d)$ 
must have empty-intersection with $A_i$. This is possible only when $B_i$ contains the set $B_{i+1}\cap \inbfs(a_i,d)$, 
since $A_i$ intersects $B_i$, but not the set $B_{i+1}\cap \inbfs(a,d)$. 
Notice that $B_{i+1}\cap \inbfs(a_i,d)$ contains 
each vertex $x\in B_{i+1}$ that satisfy $\outecc(x)\leq d$. So for any vertex $x\in B_{i+1}$, on 
substituting $d=\outecc(x)$, we get that either $\depth(\outbfs(A_i))\leq \outecc(x)$ 
or $x\in B_i$.
\end{proof}

\begin{lemma}
The size of the set $S$ is at most $O(kn^{1/k}\log n)$ and it satisfies the condition that
$\depth(\outbfs(S))$ is at most $\outecc(x)$, for each $x\in V$.
\end{lemma}
\begin{proof}
Consider any vertex $x\in V$. Let $j\in [1,k]$ be the largest index such that $x\in B_j$.
If $j=1$, then $x\in B_1\subseteq S$, and thus $\depth(\outbfs(S))\leq \depth(\outbfs(x))= \outecc(x)$.
If $j>1$, then $x\notin B_{j-1}$, and by Lemma~\ref{lemma:ecc-new-1}, 
$\depth(\outbfs(A_{j-1}))\leq \outecc(x)$, which shows that $\depth(\outbfs(S))\leq \outecc(x)$.
\end{proof}

Since $\depth(\outbfs(S))\leq \outecc(x)$, for each $x\in V$, it follows that 
$\depth(\outbfs(S))$ is bounded by $\rad(G)$. On substituting $k=\log_2 n$, we get that 
$|S|=O(\log^2 n)$, and time for computing $S$ is $O(mk)=O(m\log  n)$.
To compute a $2$-eccentricity spanner $H$, we just take union of $\inbfs(s)$ and $\outbfs(s)$, for $s\in S$.
For any $x,y\in V$, there will exists a vertex $s\in S$ satisfying $d_G(s,y)\leq \rad(G)$,
and so $d_H(x,y)\leq d_H(x,s)+d_H(s,y)\leq \outecc_G(x)+\rad(G)\leq 2\outecc_G(x)$.
Also $H$ is a $2$-radius spanner because if $c\in V$ is the vertex in $G$ with minimum eccentricity, 
then $\rad(H)\leq \outecc_H(c)\leq 2\outecc_G(c)=2\rad(G)$.
 From the above discussion, we obtain the following theorem.

\begin{theorem}
There exists an algorithm that for any directed weighted graph $G=(V,E)$ with $n$ vertices and $m$ edges, 
computes in $O(m\log^2 n)$ expected time a $2$-eccentricity spanner (and a $2$-radius spanner) of $G$
with at most $O(n\log^2n)$ edges.
\label{theorem:static-ecc-best}
\end{theorem}

\subsubsection*{Implication: An $\O(m)$ time 2-approximation of eccentricities.}
The set $S$ constructed above can also help us to obtain a $2$-approximation of graph eccentricities.
For any vertex $x\in V$, we approximate its out-eccentricity by $\outecc'(x)=\max_{s\in S}d_G(x,s)+\depth(\outbfs(S))$.
Observe $\outecc'(x)\leq \outecc_G(x)+\rad(G)\leq 2\outecc_G(x)$.
Now $\outecc'(x)\geq \outecc(x)$, because for any $y\in V$, 
if $s_y\in S$ is the vertex satisfying $d_G(s_y,y)=\depth(y,\outbfs(S))$, 
then $d_G(x,y)\leq d_G(x,s_y)+d_G(s_y,y)$ which is at most $\max_{s\in S}d_G(x,s)+\depth(\outbfs(S))=\outecc'(x)$.
Therefore, $\outecc'(x)$ is $2$-approximation of out-eccentricity of $x$. 
Observe that given the set $S$, in total $O(m\log^2n)$ time we can compute $\outecc'(x)$, for $x\in V$. We thus 
have the following theorem.

\begin{theorem}
For any directed weighted graph $G=(V,E)$ with $n$ vertices and $m$ edges, we can compute
an estimate $\outecc'(x)$, for $x\in V$, satisfying $\outecc_G(x)\leq \outecc'(x)\leq 2\outecc_G(x)$
 in $O(m\log^2 n)$ expected total time .
\label{theorem:static-ecc-estimate}
\end{theorem}

\subsection{Dynamic Maintenance of Eccentricity Spanner and Radius Spanner}

We now present our results on dynamic maintenance of eccentricity spanners.

First consider the incremental scenario.
For any vertex $w\in V$, let $q(w)$ denote the maximum integer such that 
$\inbfs(w)$ truncated to depth $q(w)$, i.e. $\inbfs(w,q(w))$, contains at most $\sqrt{n\log n}$
vertices. Observe that for any $w\in V$, we can incrementally maintaining $\inbfs(w)$, 
$\nin(w,\sqrt{n\log n})$, and $q(w)$ in a total of $O(mn)$ time, or $O(m\maxD)$ time when 
$\maxD$ is an upper bound on the diameter of $G$ throughout the run of algorithm.
Also we can dynamically maintain a set $S_{2,incr}(w)$ whose size is at most $\sqrt{n\log n}$ 
and contains $\inbfs(w,(1-\epsilon)q(w))$ as follows.
In the beginning, say at time $t_0$, we initialize $S_{2,incr}(w)=\inbfs_{t_0}(w,(1-\epsilon)q_{t_0}(w))$,
since $S_{2,incr}(w)\subseteq \inbfs_{t_0}(w,q_{t_0}(w))\subseteq \nin_{t_0}(w,\sqrt{n\log n})$,
we have $|S_{2,incr}(w)|\leq \sqrt{n\log n}$. Now we store in $\ell_0$ the value $(1-\epsilon)q_{t_0}(w)$
and keep adding all those vertices to $S_{2,incr}(w)$ whose depth in $\inbfs(w)$ reaches a value $\leq \ell_0$,
as long as $|S_{2,incr}(w)|\leq \sqrt{n\log n}$. When $|S_{2,incr}(w)|$ exceeds the value $\sqrt{n\log n}$,
then $q(w)$ must have fallen by a ratio of at least $(1-\epsilon)$, and we at that time recompute $\ell_0$
and $S_{2,incr}(w)$. Observe that between re-computations of $\ell_0$, the set $S_{2,incr}(w)$ only
grows with time. Now the total time for maintaining $S_{2,incr}(w)$ is $O(m\maxD)$;
and the number of times it is rebuild from scratch is at most $O(\epsilon^{-1}\log n)$.

Our incremental algorithm maintains a pair $(S_1,a)\in \P(V)\times V$ such that at any time instance $t$,
$\depth_t(a,\outbfs(S_1))\geq (1+\epsilon)^{-1}\depth_t(\outbfs(S_1))$.
Recall, we showed in construction of incremental dominating set-pair that
the total time for maintaining such a pair is $O(m\epsilon^{-1}\maxD\log^3n+ m\epsilon^{-2}\log^4 n)$, 
and the number of times the pair $(S_1,a)$ changes is at most $O(\epsilon^{-2}\log^4 n)$.
For a given pair $(S_1,a)$, we maintain the set $S_{2,incr}(a)$  as described above that takes $O(m\maxD)$ time.
So the total time for maintaining triplet $(S_1,a,S_{2,incr}(a))$ is $O(m\maxD\epsilon^{-2}\log^4 n)$,
and the number of times it is recomputed from scratch is $O(\epsilon^{-3}\log^5 n)$.
The incremental eccentricity spanner is just the subgraph which is union of $\inbfs(s)$ and $\outbfs(s)$
for $s\in S_1\cup S_{2,incr}(a)$. This maintenance takes $O(\epsilon^{-3}\log^5 n \cdot m\maxD\sqrt{n\log n})$
time. To prove the correctness consider any vertex $x\in V$.
Let $d$ be $\outecc_t(x)$, at some time $t$. If 
$\depth_t(\outbfs(S_1))\nleq (1+\epsilon)^2d$, then $\depth_t(a,\outbfs(S_1))\nleq (1+\epsilon)d$,
and so $\inbfs(a,(1+\epsilon)d)$ has empty intersection with $S_1$. Thus 
$\inbfs(a,(1+\epsilon)d)$ contains strictly less then $\sqrt{n\log n}$ vertices,
and so $q(a)\leq (1+\epsilon)d$.
This implies that either $\depth_t(\outbfs(S_1))\leq (1+\epsilon)^2d$, or $x\in  S_{2,incr}(a)$,
where $d=\outecc_t(x)$. On minimizing over $x$, and merging the sets $S_1\cup S_{2,incr}(a)$,
we get that at any time instance $\depth_t(\outbfs(S_1\cup S_{2,incr}(a)))\leq (1+\epsilon)^2\rad_t(G)$.
Using the same arguments as in static case, it can be shown that at any time $t$, for any $x,y\in V$
$d_{H,t}(x,y)\leq (2+3\epsilon)\outecc_t(x)$. On replacing $\epsilon$ with $\epsilon/3$,
we get the following theorem.

\begin{theorem}
For any $\epsilon\in[0,1/2]$ and any  incrementally changing graph on $n$ vertices, 
there exists an algorithm for maintaining a $(2+\epsilon)$-eccentricity spanner (and a $(2+\epsilon)$-radius spanner) 
containing at most $O(n^{3/2}\sqrt{\log n})$ edges,
whose expected amortized update time is $O((1/\epsilon^{3}) \sqrt{n}\maxD\log^{5.5} n)$,
where $\maxD$ denotes an upper bound on the maximum diameter of the graph throughout the run of the algorithm.
\label{theorem:eccentricty-spanner-incr}
\end{theorem}

Let us now focus on decremental scenario. 
Consider a time instance $t_0$. Let $S_1$ be a uniformly random subset of $V$ of size $\sqrt{n\log n}$
that intersects $\nin_{\tau}(w,\sqrt{n\log n})$, for each $w\in V$, and each time instance $\tau$.
Let $\ell_0$ be depth of $\outbfs(S_1)$ at time $t_0$. Let $t\geq t_0$ be another time instance such that
$\depth_t(\outbfs(S_1))\leq (1+\epsilon)\ell_0$. Also let $S_2=\nin_{t_0}(a,\sqrt{n\log n})$,
where $a$ is a vertex of maximum depth in tree $\outbfs(S_1)$.
Similar to arguments in Theorem~\ref{theorem:static-ecc-best}, it can be shown that at time $t_0$,
for each $x\in V$, either $\depth_{t_0}(\outbfs(S_1))\leq \outecc_{t_0}(x)$, or $x\in S_2$.
So at time $t$, for each $x\in V$, either $\depth_{t}(\outbfs(S_1))\leq (1+\epsilon)\depth_{t_0}(\outbfs(S_1))
\leq (1+\epsilon)\outecc_{t_0}(x)\leq (1+\epsilon)\outecc_{t}(x)$ (here the last inequality holds
since distances can only increase with time), or $x\in S_2$. This in turn implies that, for each $x\in V$, either 
$\depth_{t}(\outbfs(S_1))\leq  (1+\epsilon)\outecc_{t}(x)$ or $x\in S_2$.
Therefore,  $\depth_{t}(\outbfs(S_1\cup S_2))\leq  (1+\epsilon)\min_{x\in V}\outecc_{t}(x)$.
This shows that $S_1\cup S_2$ is $\langle(1+\epsilon)\rad_{t}(G),V\rangle$-dominating at time $t$.
And so union of $\inbfs(s)$ and $\outbfs(s)$ for $s\in S_1\cup S_2$ is a $(2+\epsilon)$ eccentricity spanner,
because for any $x,y\in V$, there exists an $s\in S_1\cup S_2$ such that $d_{G,t}(s,y)\leq (1+\epsilon)\rad_{t}(G)$,
and so $d_{H,t}(x,y)\leq d_{H,t}(x,s)+d_{H,t}(s,y)=d_{G,t}(x,s)+d_{G,t}(s,y)\leq \outecc_{G,t}(x)+ (1+\epsilon)\rad_{t}(G)
\leq  (2+\epsilon)\outecc_{G,t}(x)$.
Thus to dynamically maintain a $(2+\epsilon)$ eccentricity spanner, we need to recompute
$a$ and $S_2$ each time the depth of $\outbfs(S_1)$ exceeds by a factor of $(1+\epsilon)$.
Also, if $S_1\cap S_2$ in non-empty at an time, then we re-sample $S_1$, and compute $a$ and $S_2$ again.
However, expected number of re-samplings is at most $O(1)$.
As in decremental maintenance of dominating-set-pair and diameter spanner, 
it can be shown that the expected time to maintain graph $H$ is $O(\epsilon^{-1} \sqrt{n\log n}\maxD\log n)$,
so we conclude with following theorem.

\begin{theorem}
For any $\epsilon\in[0,1/2]$ and any decrementally changing graph on $n$ vertices, 
there exists an algorithm for maintaining a $(2+\epsilon)$-eccentricity spanner (and a $(2+\epsilon)$-radius spanner) 
containing at most $O(n^{3/2}\sqrt{\log n})$ edges,
whose expected amortized update time is $O((1/\epsilon) \sqrt{n}\maxD\log^{1.5} n)$,
where $\maxD$ denotes an upper bound on the maximum diameter of the graph throughout the run of the algorithm.
\label{theorem:eccentricty-spanner-decr}
\end{theorem}

\subsubsection*{Implication: Dynamic maintenance of $2$-approximate eccentricities.}
The above discussed dynamic algorithm for $2$-eccentricity spanner also imply 
a same time bound algorithm for maintaining a $2$-approximation of vertex eccentricities, because if $S$ is 
$\langle\rad(G),V\rangle$-dominating-set then for any $x\in V$, $\max_{s\in S}d_G(x,s)+\depth(\outbfs(s))$ 
is a $2$-approximation of $\outecc(x)$. Since the total time for maintaining the values 
$\max_{s\in S}d_G(x,s)+\depth(\outbfs(s))$,
for $x\in V$, is $O(m|S|\maxD)$, we obtain the following result.

\begin{theorem}
For any $\epsilon\in[0,1/2]$, there exists an incremental (and decremental) algorithm that
maintains for an $n$-vertex directed graph a $(2+\epsilon)$-approximation of graph eccentricities.
The expected amortized update time is $O((1/\epsilon^{3}) \sqrt{n}\maxD\log^{5.5} n)$ 
for incremental setting and $O((1/\epsilon) \sqrt{n}\maxD\log^{1.5} n)$ for decremental setting,	
where, $\maxD$ denotes an upper bound on the diameter of the graph throughout the run of the algorithm.
\end{theorem}

\section{Fault-Tolerant: Diameter, Diameter Spanners, and Eccentricity Spanners}
\label{section:ft}

In order to compute fault-tolerant data-structures, our first step is to compute a set $S_1$
of size $\sqrt{8n\log n}$ that has non-empty intersection with $\nin_{G\setminus x}(w,\sqrt{8n\log n}))$,
for each vertex $w\in V$, and each possible failure $x\in V\cup E$. A trivial way to even verify whether $S_1$ satisfies 
the aforesaid condition requires $O(mn^{2})$ time, since we have $n$ choices for vertex $w$, 
$n$ choices for failures in trees $\inbfs(w)$/$\outbfs(w)$, and finally computing the trees 
$\inbfs_{G\setminus x}(w)$/$\outbfs_{G\setminus x}(w)$ requires $O(m)$ time.

We first show a randomized computation of $S_1$ that takes $\O(\max\{n^{2.5},mn\})$
 time. Throughout this section 
let $r$ denote the value $\sqrt{8n\log n}$. Also let $\cal O$ denote the distance-sensitivity-oracle for 
directed graphs~\cite{DemetrescuTCR:08,BernsteinK:09} that given any $u,v\in V$ 
and $x\in V\cup E$ can output the last edge on $\pi_{G\setminus x}(u,v)$ in constant time.
This data structure can be computed in $\O(mn)$ time and takes $O(n^2 \log n)$ space.
We initialize $S_1$ to be a uniformly random subset of $V$ of size $r$.
For each $w\in V$, we compute $\inbfs(w)$ and check if $S_1$ intersects $\nin(w,r)$, if it doesn't even for a single 
vertex $w$, then we re-sample $S_1$. Next for each possible vertex failure $x\in \nin(w,r)$ (or edge failure $x\in \inbfs(w)$ with both 
end-points in $\nin(w,r)$), we compute the tree $\inbfs_{G\setminus x}(w)$.
Observe that $x$ has at most $O(r)=O(\sqrt{n})$ relevant choices, as for any other remaining option from $E\cup V$, 
the set $\nin(w,r)$ remains unaltered. Also computation of tree $\inbfs_{G\setminus x}(w)$
can be performed  in $O(n)$ time using $\cal O$. Once we have tree $\inbfs_{G\setminus x}(w)$,
we check again if $S_1$ intersects $\nin_{G\setminus x}(w,r)$, if it doesn't then we re-sample $S_1$. 
The expected number of re-samplings to compute the desired $S_1$ is $O(1)$.
Thus, the total expected time to compute $S_1$ is $\O(\max\{n^{2.5},mn\})$.

The following theorem shows construction of diameter spanner oracle that after any edge or vertex failure 
reports  a $1.5$-diameter spanner, containing at most $\O(n^{1.5})$ edges in $\O(n^{1.5})$ time.

\begin{theorem}
Any $n$-vertex directed graph $G=(V,E)$, can be preprocessed in $\O(\max\{n^{2.5},mn\})$ expected time to obtain an 
$\O(\max\{n^{2.5},mn\})$ size data structure $\cal D$ that after any any edge or vertex failure $x$, reports 
a $1.5$-diameter spanner of graph $G\setminus x$ containing at most $O(n\sqrt{n\log n})$ edges in $O(n\sqrt{n\log n})$ time. 

Moreover, given any edge $e$ and any failure $x$, the data-structure can answer the query of whether or not $e$ lies in 
$1.5$-diameter spanner of graph $G\setminus x$ in $O(\sqrt{n\log n})$ time. 
\label{theorem:ft-diam-spanner}
\end{theorem}

\begin{proof}
We compute the set $S_1$ stated in beginning of the section, tree $\outbfs(S_1)$, and a vertex $a$
having maximum depth in $\outbfs(S_1)$. 
For each edge or vertex $x$ lying in $\outbfs(S_1)$ we compute and store
(i)~the vertex $a_x$ of maximum depth in $\outbfs_{G\setminus x}(S_1)$, (ii) the set 
$\nin_{G\setminus x}(a_x,r)$.
Also for each vertex failure $x\in \nin(a,r)$ (or edge failure $x\in \inbfs(a)$ with both 
end-points in $\nin(a,r)$), we compute and store $\nin_{G\setminus x}(a,r)$.
This takes $O(nr+r^2)=O(n\sqrt{n\log n})$ space.
Next, we compute the $O(n^2 \log n)$ spaced distance sensitivity oracle $\cal O$ from \cite{DemetrescuTCR:08,BernsteinK:09}.
We assume that the edge weights in $G$ are slightly perturbed so that all distances in $G$ are distinct even after an edge/vertex 
failure. Therefore, 
(i) for any $w\in V$ and $x\in V\cup E$, in linear time oracle $\cal O$ 
can output  $\inbfs_{G\setminus x}(w)$ and $\outbfs_{G\setminus x}(w)$;
(ii) given any $w\in V$, $x\in V\cup E$, and $e\in E$, in constant time $\cal O$ 
can output  whether or not $e$ lies in $\inbfs_{G\setminus x}(w)$ and $\outbfs_{G\setminus x}(w)$.
Observe that the total pre-processing time is $\O(\max\{n^{2.5},mn\})$. 

We now explain the query process. Given a failing edge/vertex $x$, we first extract a vertex $a_0$ having maximum depth
in  $\outbfs_{G\setminus x}(S_1)$ and a set $S_2$ consisting of vertices $\nin_{G\setminus x}(a_0,r)$.
Extracting this information from $\cal D$ takes $O(r)=O(\sqrt{n\log n})$ time.
To output a $1.5$-diameter spanner we just output union of $\inbfs(s)$ and $\outbfs(s)$ for $s\in S_1\cup S_2$,
recall that these trees are computable from $\cal O$ in linear time.
Using the same arguments as in Theorem~\ref{theorem:static_diam1}, it can be shown the 
outputted graph will be a $1.5$-diameter spanner.
This takes $O(nr)=O(n\sqrt{n\log n})$ time. To verify whether or not a given edge $e$ lies in 
the $1.5$-diameter spanner, we iterate over each $s\in S_1\cup S_2$, and check whether or not $e$ lies in
$\inbfs(s)$/$\outbfs(s)$. This takes $O(\sqrt{n\log n})$ time, for any edge $e$.
Also observe that, using the same arguments as in Theorem~\ref{theorem:dynamic-diam-1}, it can be shown that 
the value $1.5(\max_{s\in (S_1\cup S_2)} \max\{\inecc(s),\outecc(s)\})$ is a 1.5-approximation of the diameter
of graph $G\setminus x$.  
\end{proof}

Next we present our diameter-sensitivity-oracle. Observe that a trivial diameter-sensitivity-oracle 
would be to compute a $(D/2,\lfloor D/2\rfloor)$ dominating set-pair $(S_1,S_2)$ of size bound $\langle r,r\rangle$, 
and a 1.5-diameter spanner $H$ which is union of $\inbfs(s)$ and $\outbfs(s)$, for $s\in S_1\cup S_2$. 
If a failure $x$ is not in $H$, then $1.5(\max_{s\in (S_1\cup S_2)} \max\{\inecc(s),\outecc(s)\})$ would still be 
a 1.5-diameter approximation of graph $G\setminus x$. If a failure $x$ lies in $H$, then it has at most $O(n\sqrt{n\log n})$
choices, and for each of possible choice we can compute and store again a 1.5-diameter-approximation in
$O(m\sqrt{n\log n})$ time. Thus total time for this procedure is $O(mn^2 \log n)$.
Now from Theorem~\ref{theorem:ft-diam-spanner}, $G$ can can be preprocessed in $\O(\max\{n^{2.5},mn\})$ 
expected time to compute a data-structure $\cal D$ that given any edge or vertex failure $x\in H$, 
computes in $O(n\sqrt{n\log n})$ time 
a 1.5-diameter spanner of $G\setminus x$. Moreover, we also showed that in the same time it can compute a 1.5-approximation 
of the diameter of graph $G\setminus x$. Since there are $O(n\sqrt{n \log n})$ choices for $x$, 
and for each such choice it takes $O(n\sqrt{n\log n})$ time to compute a 1.5-diameter-approximation,
the total time of this process is $\O(n^3)$. We thus conclude with following theorem.

\begin{theorem}
Any $n$-vertex directed graph $G=(V,E)$, can be preprocessed in $\O(n^3)$ expected time to 
obtain an $O(n\sqrt{n\log n})$ size data structure $\cal D$ that after any any edge or vertex failure $x$, reports 
a $1.5$-approximation of diameter of graph $G\setminus x$ in constant time.
\label{theorem:ft-diam}
\end{theorem}

The data-structure for fault-tolerant-eccentricity spanner is exactly similar to diameter spanner data structure 
from Theorem~\ref{theorem:ft-diam-spanner}. The proof of correctness follows from the fact that in 
Theorem~\ref{theorem:ft-diam-spanner}, we essentially after any failure $x$, first compute a valid-set-pair 
$(S_1,S_2)$ for $G\setminus x$, and next output union of shortest-path-trees in $G\setminus x$ rooted at 
vertices in $S_1\cup S_2$. Using the arguments in Theorem~, it can be shown that 
$S_1\cup S_2$ is $\langle \rad(G\setminus x),V\rangle$-dominating, and therefore, the outputted graph is also a 
$2$-eccentricity spanner.

\begin{theorem}
Any $n$-vertex directed graph $G=(V,E)$, can be preprocessed in $\O(\max\{n^{2.5},mn\})$ expected time to obtain an 
$\O(\max\{n^{2.5},mn\})$ size data structure $\cal D$ that after any any edge or vertex failure $x$, reports 
a $2$-eccentricity spanner of graph $G\setminus x$ containing at most $O(n\sqrt{n\log n})$ edges in $O(n\sqrt{n\log n})$ time. 

Moreover, given any edge $e$ and any failure $x$, the data-structure can answer the query of whether or not $e$ lies in 
$2$-eccentricity spanner of graph $G\setminus x$ in $O(\sqrt{n\log n})$ time. 
\label{theorem:ft-ecc-spanner}
\end{theorem}

\section{Lower Bounds for Diameter Spanners and Eccentricity Spanners}
\label{section:lower-bounds}

In this section, we prove lower bounds on the number of edges in diameter spanners and eccentricity spanners.
In particular, we will prove that our constructions for $1.5$-diameter spanner, $5/3$-diameter spanner, and
$2$-eccentricity spanner are tight, for graphs with low diameter.

\subsection{Lower bound for $1.5$-Diameter Spanner}

\begin{theorem}
For every $n$ and every $t$, there exists an $n$-vertex directed graph $G=(V,E)$ with diameter $2(t+1)$ 
such that any subgraph $H=(V,E'\subseteq E)$ of $G$ with 
$\dia(H)\leq 3t+1$ contains $\Omega(n^2/t^2)$ edges.
\end{theorem}

\begin{proof}
Let $N$ be such that $n=N(2t+2)$. The construction of $G$ is as follows.
The vertex set $V(G)$ comprises of four sets $A,B,C,D$ respectively of size $tN$, $N$, $N$, and $tN$. 
The vertices in $A$ are denoted by $a_{k,i}$ where $k\in[1,t]$ and $i\in[1,N]$.
The vertices in $B$ are denoted by $b_i$ where $i\in[1,N]$.
The vertices in $C$ are denoted by $c_j$ where $j\in[1,N]$.
The vertices in $D$ are denoted by $d_{k,j}$ where $k\in[1,t]$ and $j\in[1,N]$.
The edges in $G$ are as follows:
(i) each vertex $a_{k,i} \in A$ has one out-going edge, namely $(a_{k,i},a_{k+1,i})$ if $k<t$
and $(a_{k,i},b_i)$ when $k=t$;
(ii) between sets $B$ and $C$ there is a complete bipartite graph, that is, each $(b_i,c_j)$
is an edge;
(iii) each vertex $d_{k,j} \in D$ has one incoming edge, namely $(d_{k-1,j},d_{k,j})$ if $k>1$
and $(c_j,d_{k,j})$ when $k=1$;
(iv) for each $x\in B\cup C\cup D$ and each $a_{1,i}\in A$, there is an edge $(x,a_{1,i})$  in $G$.
See Figure~\ref{Figure:lb_3by2}.\\[-3mm]

\begin{figure}[bth]
\centering
\includegraphics[scale=.4]{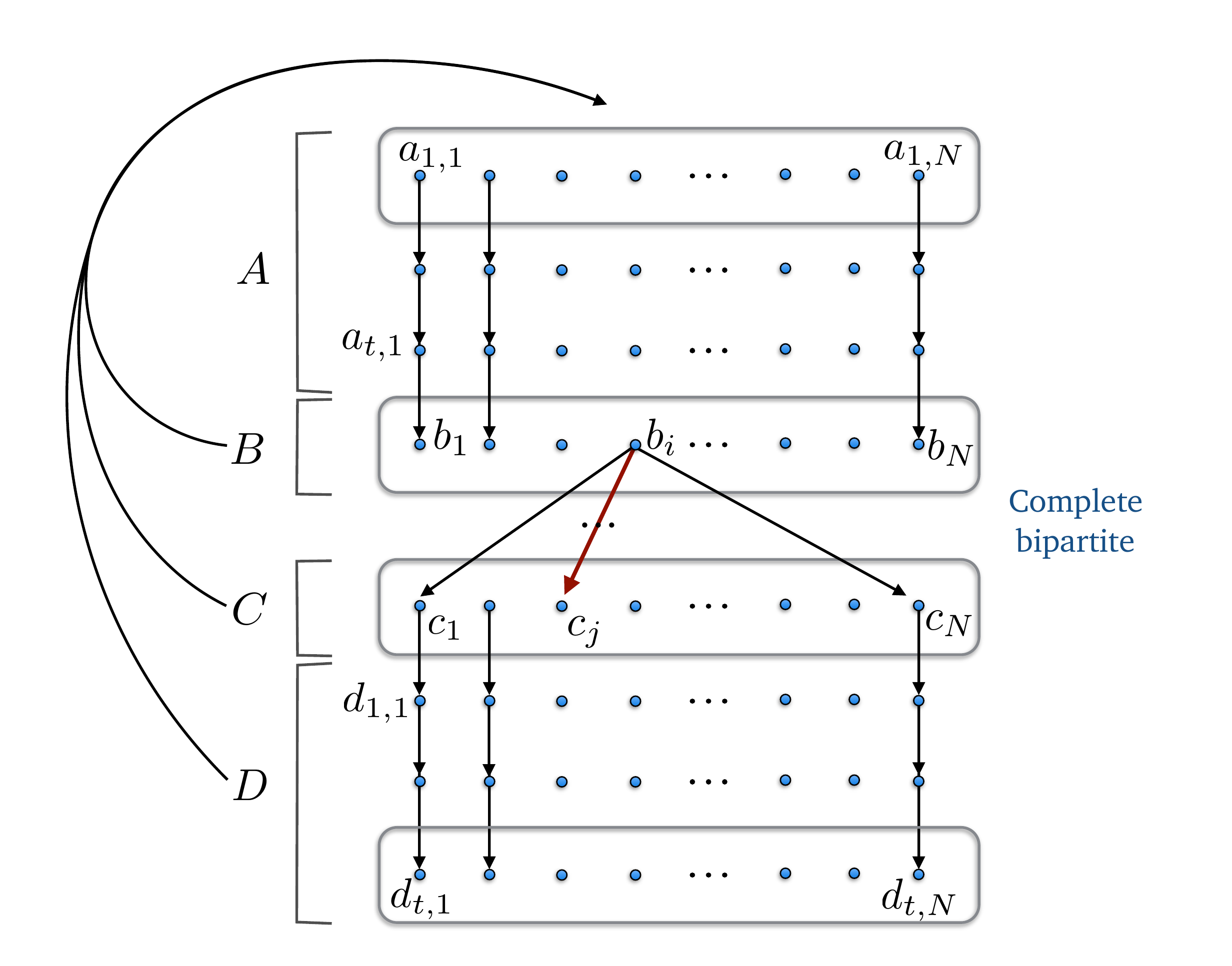}
\caption{
Illustration of lower bound for diameter spanners with 1.5 stretch.
}
\label{Figure:lb_3by2}
\end{figure}

\noindent
We will show that the diameter of $G$ is at most $2(t+1)$.
\begin{itemize}
\item In order to focus on vertex pairs in product $(B\cup C\cup D)\times V(G)$, consider any vertex $x\in B\cup C\cup D$. 
\begin{itemize}
\item For $a_{k,i}\in A$, $(x,a_{1,i},a_{2,i},\ldots,a_{k,i})$ is a path  of length at most $2t+2$.
\item For $b_{i}\in B$, $(x,a_{1,i},a_{2,i},\ldots,a_{t,i},b_i)$ is a path of length at most $2t+2$.
\item For $c_{j}\in C$, $(x,a_{1,j},a_{2,j},\ldots,a_{t,j},b_j,c_j)$ is a path of length at most $2t+2$.
\item For $d_{k,j}\in D$, $(x,a_{1,j},a_{2,j},\ldots,a_{t,j},b_j,c_j,d_{1,j},d_{2,j},\ldots,d_{k,j})$ is a path of length at most $2t+2$.
\end{itemize}
\item Next consider any vertex $a_{k,i}\in A$.
\begin{itemize}
\item For $a_{k',i'}\in A$, $(a_{k,i},a_{k+1,i},\ldots,a_{t,i},b_i,a_{1,i'},a_{2,i'},\ldots,a_{k',i'})$ is a path of length at most $2t+2$.
\item For $b_{i'}\in B$, $(a_{k,i},a_{k+1,i},\ldots,a_{t,i},b_i,a_{1,i'},a_{2,i'},\ldots,a_{t,i'},b_{i'})$ is a path of length at most $2t+2$.
\item For $c_{j}\in C$, $(a_{k,i},a_{k+1,i},\ldots,a_{t,i},b_i,c_j)$ is a path of length at most $2t+2$.
\item For $d_{k,j}\in D$, $(a_{k,i},a_{k+1,i},\ldots,a_{t,i},b_i,c_j,d_{1,j},d_{2,j},\ldots,d_{k,j})$ is a path of length at most $2t+2$.
\end{itemize}

\end{itemize}

To verify that the  diameter of $G$ is exactly $2t+2$, observe that the distance between vertices
$d_{t,1}$ and $d_{t,N}$ in $G$ is equal to $2t+2$.

Observe that on removal of any edge $(b_i,c_j)\in (B\times C)$ from $G$, 
the distance between $a_{1,i}$ and $d_{t,j}$ becomes $3t+2$; this is because
any shortest path from $a_{1,i}$ to $d_{t,j}$ in $G\setminus (b_i,c_j)$ has form
$(a_{1,i},a_{2,i},\ldots,a_{t,i},b_i,a_{1,i'},a_{2,i'},\ldots,a_{t,i'},b_{i'},c_j,$ $d_{1,j},d_{2,j},\ldots,d_{t,j})$,
were $i'\neq i$.
Therefore, any subgraph $H$ of $G$ whose diameter is at most $3t+1$ must contain all the edges lying in 
the set $B\times C$,
that is, $H$ should have at least $N^2=n^2/(2t+2)^2= \Omega(n^2/t^2)$ edges.
This completes our proof.
\end{proof}

\begin{corollary}[\normalfont{Reminder of Theorem~\ref{theorem:3/2-diam-up-lb}(b)}]
For every $n$ and every $D\leq (n^{1/4})$, there exists a unweighted directed graph $G$ 
with $\Theta(n)$ vertices and diameter $D$, such that any subgraph $H$ of $G$ that satisfies
$\dia(H)$ is strictly less than $1.5\hspace{1mm}\dia(G)-1$ contains at least $\Omega(n^{1.5})$ edges.
\end{corollary}


\subsection{Lower Bound for $2$-Eccentricity Spanner}

\begin{theorem}
For every $n$ and every $t$, there exists an $n$-vertex directed graph $G=(V,E)$ and a
subset $S$ of $V$ of size $\Theta(n/t)$ having eccentricities $(t+1)$ in $G$, such that
any subgraph $H=(V,E'\subseteq E)$ of $G$ satisfying $\outecc_H(s)< 2\outecc_G(s)$, for each $s\in S$,
contains $\Omega(n^2/t^2)$ edges.
\end{theorem}

\begin{proof}
The construction of graph $G$ remains same as that in Figure~\ref{Figure:lb_3by2}.
So $N$ satisfies $n=N(2t+2)$ vertices. Now the set $S$ is the set $B$ in the graph $G$.
Observe that eccentricity of each $b_i\in B$ is exactly $t+1$. Now on removing
any edge $e=(B_i,c_j)\in (B \times C)$ from $G$, the distance between $b_i$
and $d_{t,j}$ becomes $2(t+1)$. This shows that any subgraph $H$ of $G$
satisfying $\outecc_H(s)< 2\outecc_G(s)$, for each $s\in S$, must contain 
all the edges of the set $B\times C$, and hence must have  $\Omega(n^2/t^2)$ edges.
\end{proof}

\begin{corollary}[\normalfont{Reminder of Theorem~\ref{theorem:2-ecc-up-lb}(b)}]
For every $n$ and every $R$, there exists a unweighted directed graph $G$ with $\Theta(n)$ vertices and radius $R$, 
such that any subgraph $H$ of $G$ that is a $(2-\epsilon)$-eccentricity spanner
of $H$, for any $\epsilon>0$, contains at least $\Omega(n^2/R^2)$ edges.
\end{corollary}


\subsection{Lower Bound for $5/3$-Diameter Spanner}

We present below our lower-bound construction for $5/3$-diameter spanner.

\begin{theorem}
For every $n$ and every $t$, there exists an $n$-vertex directed graph $G=(V,E)$ with diameter $3(t+1)$ 
such that any subgraph $H=(V,E'\subseteq E)$ of $G$ with 
$\dia(H)\leq 5t+1$ contains $\Omega(n^{1.5}/t^{1.5})$ edges.
\end{theorem}

\begin{proof}
Let $N$ be such that $n=N^2(3t+3)$. The construction of $G$ is as follows.
The vertex set $V(G)$ comprises of four sets $A,B,C$ respectively of size $N^2(t+1)$. 
The vertices in $A$ are denoted by $a_{k,i,j}$ where $k\in[0,t]$ and $i,j\in[1,N]$.
The vertices in $B$ are denoted by $b_{k,i,j}$ where $k\in[0,t]$ and $i,j\in[1,N]$.
The vertices in $C$ are denoted by $c_{k,i,j}$ where $k\in[0,t]$ and $i,j\in[1,N]$.
The edges in $G$ are as follows:
(i) each vertex $a_{k,i,j} \in A$ has an out-going edge, namely $(a_{k,i,j},a_{k+1,i,j})$ if $k<t$;
(ii) each vertex $b_{k,i,j} \in B$ has an out-going edge, namely $(b_{k,i,j},b_{k+1,i,j})$ if $k<t$;
(iii) each vertex $c_{k,i,j} \in C$ has an out-going edge, namely $(c_{k,i,j},c_{k+1,i,j})$ if $k<t$;
(iv) each vertex $a_{t,i,j} \in A$ has an out-going edge to $b_{0,i',j'}\in B$ whenever $j=j'$;
(v) each vertex $b_{t,i,j} \in B$ has an out-going edge to $c_{0,i',j'}\in C$ whenever  $i=i'$;
(vi) each vertex $a_{t,i,j} \in A$ has an out-going edge to $a_{0,i',j'}$ whenever $j=j'$;
(vii) each vertex $b_{t,i,j} \in B$ has an out-going edge to $b_{0,i',j'}$ whenever $i=i'$ or $j=j'$;
(viii) each vertex $x \in (\cup_{i,j\in [1,N]} b_{t,i,j}) \cup C$ has an out-going edge to $a_{0,i',j'}$ whenever $i=i'$ or $j=j'$.\\

\noindent
We will show that the diameter of $G$ is at most $3(t+1)$ by doing a case by case analysis.
\begin{itemize}

\item Consider pair $(a_{k,i,j},a_{k',i',j'})\in A\times A$. 
\begin{itemize}
\item $(a_{k,i,j},a_{k+1,i,j},\ldots,a_{t,i,j},b_{0,i',j},\ldots,b_{t,i',j},a_{0,i',j'},\ldots,a_{k',i',j'})$ is a path  of length $\leq 3t+3$.
\end{itemize}

\item Consider pair $(a_{k,i,j},b_{k',i',j'})\in A\times B$. 
\begin{itemize}
\item $(a_{k,i,j},a_{k+1,i,j},\ldots,a_{t,i,j},b_{0,i',j},\ldots,b_{t,i',j},b_{0,i',j'},\ldots,b_{k',i',j'})$ is a path  of length $\leq 3t+3$.
\end{itemize}

\item Consider pair $(a_{k,i,j},c_{k',i',j'})\in A\times C$. 
\begin{itemize}
\item $(a_{k,i,j},a_{k+1,i,j},\ldots,a_{t,i,j},b_{0,i',j},\ldots,b_{t,i',j},c_{0,i',j'},\ldots,c_{k',i',j'})$ is a path  of length $\leq 3t+3$.
\end{itemize}\vspace{3mm}

\item Consider pair $(b_{k,i,j},a_{k',i',j'})\in B\times A$. 
\begin{itemize}
\item $(b_{k,i,j},\ldots,b_{t,i,j},c_{0,i,j'},a_{0,i',j'},\ldots,a_{k',i',j'})$ is a path  of length $\leq 2t+2$.
\end{itemize}

\item Consider pair $(b_{k,i,j},b_{k',i',j'})\in B\times B$. 
\begin{itemize}
\item $(b_{k,i,j},\ldots,b_{t,i,j},b_{0,i',j},\ldots,b_{t,i',j},b_{0,i',j'},\ldots,b_{k',i',j'})$ is a path  of length $\leq 3t+3$.
\end{itemize}

\item Consider pair $(b_{k,i,j},c_{k',i',j'})\in B\times C$. 
\begin{itemize}
\item $(b_{k,i,j},\ldots,b_{t,i,j},b_{0,i',j},\ldots,b_{t,i',j},c_{0,i',j'},\ldots,c_{k',i',j'})$ is a path  of length $\leq 3t+3$.
\end{itemize}\vspace{3mm}

\item Consider pair $(c_{k,i,j},a_{k',i',j'})\in C\times A$. 
\begin{itemize}
\item $(c_{k,i,j},a_{0,i,j'},\ldots,a_{t,i,j'},a_{0,i',j'},\ldots,a_{k',i',j'})$ is a path  of length $\leq 2t+2$.
\end{itemize}

\item Consider pair $(c_{k,i,j},b_{k',i',j'})\in C\times B$. 
\begin{itemize}
\item $(c_{k,i,j},a_{0,i,j'},\ldots,a_{t,i,j'},b_{0,i',j'},\ldots,b_{k',i',j'})$ is a path  of length $\leq 2t+2$.
\end{itemize}

\item Consider pair $(c_{k,i,j},c_{k',i',j'})\in C\times C$. 
\begin{itemize}
\item $(c_{k,i,j},a_{0,i,j'},\ldots,a_{t,i,j'},b_{0,i',j'},\ldots,b_{t,i',j'},c_{0,i',j'},\ldots,c_{k',i',j'})$ is a path  of length $\leq 3t+3$.
\end{itemize}
\end{itemize}

Consider any six indices $i_x,j_x,i_y,j_y\in [1,N]$ such that $i_x\neq i_y$ and $j_x\neq j_y$.
Let $x$ denote the vertex $a_{0,i_x,j_x}$ and $y$ denote the vertex $c_{t,i_y,j_y}$.
Let $G_0$ be the graph obtained from $G$ by removing following edges: 
$(a_{t,i_x,j_x},b_{0,i_y,j_x})$, $(a_{t,i_x,j_x},a_{0,i_y,j_x})$, 
and $(b_{t,i_y,j_x},c_{0,i_y,j_y})$.
One can verify that distance between vertices $x$ and $y$ in $G_0$ is exactly $5t+4$.

Therefore, any subgraph $H$ of $G$ whose diameter is at most $5t+3$ must contain  
for each $i_x,j_x,i_y,j_y\in [1,N]~(i_x\neq i_y,~j_x\neq j_y)$ either of the three edges: 
$(a_{t,i_x,j_x},b_{0,i_y,j_x})$, $(a_{t,i_x,j_x},a_{0,i_y,j_x})$, or $(b_{t,i_y,j_x},c_{0,i_y,j_y})$.
This shows that $H$ must contain $\Omega(N\sqrt{N})=\Omega(n^{1.5}/t^{1.5})$ edges.
\end{proof}

\begin{corollary}
 For every $n$ and every $D\leq (n^{1/11})$, there exists a unweighted directed graph $G$ 
with $\Theta(n)$ vertices and diameter $D$, such that any subgraph $H$ of $G$ for which
$\dia(H)$ is strictly less than $(5/3\hspace{1mm}\dia(G)-1)$ contains at least $\Omega(n^{4/3}D^{1/3})$ edges.
\end{corollary}

\subsection*{Acknowledgements}
The authors are thankful to Liam Roditty for helpful discussions.

\bibliography{references}

\begin{thebibliography}{10}

\bibitem{AB16}
Amir Abboud and Greg Bodwin.
\newblock The 4/3 additive spanner exponent is tight.
\newblock In {\em Proceedings of the 48th Annual ACM Symposium on Theory of
  Computing (STOC)}, pages 351--361, 2016.
\newblock URL: \url{http://doi.acm.org/10.1145/2897518.2897555}, \href
  {http://dx.doi.org/10.1145/2897518.2897555}
  {\path{doi:10.1145/2897518.2897555}}.

\bibitem{ACIM99}
D.~Aingworth, C.~Chekuri, P.~Indyk, and R.~Motwani.
\newblock Fast estimation of diameter and shortest paths (without matrix
  multiplication).
\newblock {\em SIAM Journal on Computing}, 28(4):1167--1181, 1999.
\newblock URL: \url{https://doi.org/10.1137/S0097539796303421}, \href
  {http://dx.doi.org/10.1137/S0097539796303421}
  {\path{doi:10.1137/S0097539796303421}}.

\bibitem{ADDJS93}
Ingo Alth{\"o}fer, Gautam Das, David Dobkin, Deborah Joseph, and Jos{\'e}
  Soares.
\newblock On sparse spanners of weighted graphs.
\newblock {\em Discrete {\&} Computational Geometry}, 9(1), 1993.

\bibitem{ALWWF03}
K.~Alzoubi, X.~Y. Li, Y.~Wang, P.~J. Wan, and O.~Frieder.
\newblock Geometric spanners for wireless ad hoc networks.
\newblock {\em IEEE Transactions on Parallel and Distributed Systems},
  14(4):408--421, 2003.

\bibitem{arxiv:diam}
Bertie Ancona, Monika Henzinger, Liam Roditty, Virginia~Vassilevska Williams,
  and Nicole Wein.
\newblock Algorithms and hardness for diameter in dynamic graphs.
\newblock {\em CoRR}, abs/1811.12527, 2018.
\newblock URL: \url{http://arxiv.org/abs/1811.12527}.

\bibitem{AusielloIMN:91}
Giorgio Ausiello, Giuseppe~F. Italiano, Alberto Marchetti{-}Spaccamela, and
  Umberto Nanni.
\newblock Incremental algorithms for minimal length paths.
\newblock {\em J. Algorithms}, 12(4):615--638, 1991.

\bibitem{Awerbuch85}
Baruch Awerbuch.
\newblock Communication-time trade-offs in network synchronization.
\newblock In {\em Proceedings of the 4th Annual ACM Symposium on Principles of
  Distributed Computing (PODC)}, pages 272--276, New York, NY, USA, 1985. ACM.
\newblock URL: \url{http://doi.acm.org/10.1145/323596.323621}, \href
  {http://dx.doi.org/10.1145/323596.323621} {\path{doi:10.1145/323596.323621}}.

\bibitem{Backurs18}
Arturs Backurs, Liam Roditty, Gilad Segal, Virginia~Vassilevska Williams, and
  Nicole Wein.
\newblock Towards tight approximation bounds for graph diameter and
  eccentricities.
\newblock In {\em Proceedings of the 50th Annual ACM SIGACT Symposium on Theory
  of Computing (STOC)}, pages 267--280, 2018.
\newblock URL: \url{http://doi.acm.org/10.1145/3188745.3188950}, \href
  {http://dx.doi.org/10.1145/3188745.3188950}
  {\path{doi:10.1145/3188745.3188950}}.

\bibitem{BaswanaHS:07}
Surender Baswana, Ramesh Hariharan, and Sandeep Sen.
\newblock Improved decremental algorithms for maintaining transitive closure
  and all-pairs shortest paths.
\newblock {\em J. Algorithms}, 62(2):74--92, 2007.

\bibitem{BKMP10}
Surender Baswana, Telikepalli Kavitha, Kurt Mehlhorn, and Seth Pettie.
\newblock Additive spanners and ($\alpha$, $\beta$)-spanners.
\newblock {\em ACM Trans. Algorithms}, 7(1):5:1--5:26, December 2010.
\newblock URL: \url{http://doi.acm.org/10.1145/1868237.1868242}, \href
  {http://dx.doi.org/10.1145/1868237.1868242}
  {\path{doi:10.1145/1868237.1868242}}.

\bibitem{BS03}
Surender Baswana and Sandeep Sen.
\newblock A simple linear time algorithm for computing a (2k - 1)-spanner of
  o(n1+1/k) size in weighted graphs.
\newblock In {\em Proceedings of the 30th International Conference on Automata,
  Languages and Programming (ICALP)}, pages 384--396, 2003.
\newblock URL: \url{http://dl.acm.org/citation.cfm?id=1759210.1759250}.

\bibitem{BS06}
Surender Baswana and Sandeep Sen.
\newblock Approximate distance oracles for unweighted graphs in expected o(n2)
  time.
\newblock {\em ACM Trans. Algorithms}, 2(4):557--577, 2006.
\newblock URL: \url{http://doi.acm.org/10.1145/1198513.1198518}, \href
  {http://dx.doi.org/10.1145/1198513.1198518}
  {\path{doi:10.1145/1198513.1198518}}.

\bibitem{BernsteinK:09}
Aaron Bernstein and David~R. Karger.
\newblock A nearly optimal oracle for avoiding failed vertices and edges.
\newblock In {\em Proceedings of the 41st Annual {ACM} Symposium on Theory of
  Computing, {STOC} 2009, Bethesda, MD, USA, May 31 - June 2, 2009}, pages
  101--110, 2009.

\bibitem{Bodwin17}
Greg Bodwin.
\newblock Linear size distance preservers.
\newblock In {\em Proceedings of the 28th Annual ACM-SIAM Symposium on Discrete
  Algorithms (SODA)}, pages 600--615, 2017.
\newblock URL: \url{http://dl.acm.org/citation.cfm?id=3039686.3039725}.

\bibitem{CairoGR16}
Massimo Cairo, Roberto Grossi, and Romeo Rizzi.
\newblock New bounds for approximating extremal distances in undirected graphs.
\newblock In {\em Proceedings of the Twenty-Seventh Annual {ACM-SIAM} Symposium
  on Discrete Algorithms, {SODA} 2016, Arlington, VA, USA, January 10-12,
  2016}, pages 363--376, 2016.

\bibitem{Chechik2013}
Shiri Chechik.
\newblock New additive spanners.
\newblock In {\em Proceedings of the 24th Annual ACM-SIAM Symposium on Discrete
  Algorithms (SODA)}, pages 498--512, 2013.
\newblock URL: \url{http://dl.acm.org/citation.cfm?id=2627817.2627853}.

\bibitem{ChechikLRSTW:14}
Shiri Chechik, Daniel~H. Larkin, Liam Roditty, Grant Schoenebeck, Robert~Endre
  Tarjan, and Virginia~Vassilevska Williams.
\newblock Better approximation algorithms for the graph diameter.
\newblock In {\em Proceedings of the Twenty-Fifth Annual {ACM-SIAM} Symposium
  on Discrete Algorithms, {SODA} 2014, Portland, Oregon, USA, January 5-7,
  2014}, pages 1041--1052, 2014.

\bibitem{Cohen98}
E.~Cohen.
\newblock Fast algorithms for constructing t-spanners and paths with stretch t.
\newblock {\em SIAM Journal on Computing}, 28(1):210--236, 1998.
\newblock URL: \url{https://doi.org/10.1137/S0097539794261295}, \href
  {http://dx.doi.org/10.1137/S0097539794261295}
  {\path{doi:10.1137/S0097539794261295}}.

\bibitem{Cohen00}
Edith Cohen.
\newblock Polylog-time and near-linear work approximation scheme for undirected
  shortest paths.
\newblock {\em J. ACM}, 47(1):132--166, 2000.
\newblock URL: \url{http://doi.acm.org/10.1145/331605.331610}, \href
  {http://dx.doi.org/10.1145/331605.331610} {\path{doi:10.1145/331605.331610}}.

\bibitem{CSMOC15}
D.~Coleman, I.~A. Şucan, M.~Moll, K.~Okada, and N.~Correll.
\newblock Experience-based planning with sparse roadmap spanners.
\newblock In {\em 2015 IEEE International Conference on Robotics and Automation
  (ICRA)}, pages 900--905, 2015.

\bibitem{Cowen01}
Lenore~J Cowen.
\newblock Compact routing with minimum stretch.
\newblock {\em J. Algorithms}, 38(1):170--183, 2001.
\newblock URL: \url{http://dx.doi.org/10.1006/jagm.2000.1134}, \href
  {http://dx.doi.org/10.1006/jagm.2000.1134}
  {\path{doi:10.1006/jagm.2000.1134}}.

\bibitem{Cowen04}
Lenore~J. Cowen and Christopher~G. Wagner.
\newblock Compact roundtrip routing in directed networks.
\newblock {\em Journal of Algorithms}, 50(1):79 -- 95, 2004.
\newblock URL:
  \url{http://www.sciencedirect.com/science/article/pii/S0196677403001275},
  \href {http://dx.doi.org/https://doi.org/10.1016/j.jalgor.2003.08.001}
  {\path{doi:https://doi.org/10.1016/j.jalgor.2003.08.001}}.

\bibitem{PS89}
Peleg David and Sch{\"a}ffer~Alejandro A.
\newblock Graph spanners.
\newblock {\em Journal of Graph Theory}, 13(1):99--116, 1989.
\newblock URL:
  \url{https://onlinelibrary.wiley.com/doi/abs/10.1002/jgt.3190130114}, \href
  {http://dx.doi.org/10.1002/jgt.3190130114}
  {\path{doi:10.1002/jgt.3190130114}}.

\bibitem{DemetrescuTCR:08}
Camil Demetrescu, Mikkel Thorup, Rezaul~Alam Chowdhury, and Vijaya
  Ramachandran.
\newblock Oracles for distances avoiding a failed node or link.
\newblock {\em {SIAM} J. Comput.}, 37(5):1299--1318, 2008.

\bibitem{DB14}
Andrew Dobson and Kostas~E. Bekris.
\newblock Sparse roadmap spanners for asymptotically near-optimal motion
  planning.
\newblock {\em The International Journal of Robotics Research}, 33(1):18--47,
  2014.
\newblock URL: \url{https://doi.org/10.1177/0278364913498292}, \href
  {http://dx.doi.org/10.1177/0278364913498292}
  {\path{doi:10.1177/0278364913498292}}.

\bibitem{Elkin01}
Michael Elkin.
\newblock Computing almost shortest paths.
\newblock In {\em Proceedings of the 20th Annual ACM Symposium on Principles of
  Distributed Computing (PODC)}, pages 53--62, 2001.
\newblock URL: \url{http://doi.acm.org/10.1145/383962.383983}, \href
  {http://dx.doi.org/10.1145/383962.383983} {\path{doi:10.1145/383962.383983}}.

\bibitem{ElkinPeleg01}
Michael Elkin and David Peleg.
\newblock Approximating k-spanner problems for $k>2$.
\newblock In {\em Proceedings of the 8th International IPCO Conference on
  Integer Programming and Combinatorial Optimization}, pages 90--104, 2001.
\newblock URL: \url{http://dl.acm.org/citation.cfm?id=645590.659944}.

\bibitem{ElkinP01}
Michael Elkin and David Peleg.
\newblock The client-server 2-spanner problem with applications to network
  design.
\newblock In {\em Proceedings of the 8th International Colloquium on Structural
  Information and Communication Complexity (SIROCCO)}, pages 117--132, 2001.

\bibitem{EP04}
Michael Elkin and David Peleg.
\newblock $(1 + \epsilon,\beta)$-spanner constructions for general graphs.
\newblock {\em SIAM J. Comput.}, 33(3):608--631, March 2004.
\newblock URL: \url{http://dx.doi.org/10.1137/S0097539701393384}, \href
  {http://dx.doi.org/10.1137/S0097539701393384}
  {\path{doi:10.1137/S0097539701393384}}.

\bibitem{Erdo64}
P.~Erdős.
\newblock Extremal problems in graph theory.
\newblock In {\em In Theory Of Graphs and its Applications (Proc. Sympos.
  Smolenice)}, pages 29--36, 1964.

\bibitem{ES:81}
S.~Even and Y.~Shiloach.
\newblock An on-line edge deletion problem.
\newblock {\em J. ACM}, 28(1):1--4, 1981.

\bibitem{GZ10}
Jie Gao and Dengpan Zhou.
\newblock The emergence of sparse spanners and greedy well-separated pair
  decomposition.
\newblock In {\em Proceedings of the 12th Scandinavian Symposium and Workshops
  on Algorithm Theory (SWAT)}, pages 50--61, 2010.

\bibitem{GS11}
Cyril Gavoille and Christian Sommer.
\newblock Sparse spanners vs. compact routing.
\newblock In {\em Proceedings of the 23rd Annual ACM Symposium on Parallelism
  in Algorithms and Architectures (SPAA)}, pages 225--234, 2011.
\newblock URL: \url{http://doi.acm.org/10.1145/1989493.1989526}, \href
  {http://dx.doi.org/10.1145/1989493.1989526}
  {\path{doi:10.1145/1989493.1989526}}.

\bibitem{haynes1998fundamentals}
T.W. Haynes, S.~Hedetniemi, and P.~Slater.
\newblock {\em Fundamentals of Domination in Graphs}.
\newblock Chapman \& Hall/CRC Pure and Applied Mathematics. Taylor \& Francis,
  1998.
\newblock URL: \url{https://books.google.co.il/books?id=Bp9fot\_HyL8C}.

\bibitem{henning2014total}
M.A. Henning and A.~Yeo.
\newblock {\em Total Domination in Graphs}.
\newblock Springer Monographs in Mathematics. Springer New York, 2014.
\newblock URL: \url{https://books.google.co.il/books?id=BYe4BAAAQBAJ}.

\bibitem{HolzerPRW:14}
Stephan Holzer, David Peleg, Liam Roditty, and Roger Wattenhofer.
\newblock Distributed 3/2-approximation of the diameter.
\newblock In {\em Distributed Computing - 28th International Symposium, {DISC}
  2014, Austin, TX, USA, October 12-15, 2014. Proceedings}, pages 562--564,
  2014.

\bibitem{PelegU89a}
David Peleg and Jeffrey~D. Ullman.
\newblock An optimal synchronizer for the hypercube.
\newblock {\em {SIAM} J. Comput.}, 18(4):740--747, 1989.

\bibitem{PU89}
David Peleg and Eli Upfal.
\newblock A trade-off between space and efficiency for routing tables.
\newblock {\em J. ACM}, 36(3):510--530, 1989.
\newblock URL: \url{http://doi.acm.org/10.1145/65950.65953}, \href
  {http://dx.doi.org/10.1145/65950.65953} {\path{doi:10.1145/65950.65953}}.

\bibitem{RTZ02}
Liam Roditty, Mikkel Thorup, and Uri Zwick.
\newblock Roundtrip spanners and roundtrip routing in directed graphs.
\newblock In {\em Proceedings of the 13th Annual ACM-SIAM Symposium on Discrete
  Algorithms (SODA)}, pages 844--851, 2002.
\newblock URL: \url{http://dl.acm.org/citation.cfm?id=545381.545491}.

\bibitem{RodittyTZ05}
Liam Roditty, Mikkel Thorup, and Uri Zwick.
\newblock Deterministic constructions of approximate distance oracles and
  spanners.
\newblock In {\em Automata, Languages and Programming, 32nd International
  Colloquium, {ICALP} 2005, Lisbon, Portugal, July 11-15, 2005, Proceedings},
  pages 261--272, 2005.

\bibitem{RodittyW:13}
Liam Roditty and Virginia~Vassilevska Williams.
\newblock Fast approximation algorithms for the diameter and radius of sparse
  graphs.
\newblock In {\em Symposium on Theory of Computing Conference, STOC'13, Palo
  Alto, CA, USA, June 1-4, 2013}, pages 515--524, 2013.

\bibitem{TZ01}
Mikkel Thorup and Uri Zwick.
\newblock Compact routing schemes.
\newblock In {\em Proceedings of the 13th Annual ACM Symposium on Parallel
  Algorithms and Architectures (SPAA)}, pages 1--10, 2001.

\bibitem{TZ05}
Mikkel Thorup and Uri Zwick.
\newblock Approximate distance oracles.
\newblock {\em J. ACM}, 52(1):1--24, 2005.
\newblock URL: \url{http://doi.acm.org/10.1145/1044731.1044732}, \href
  {http://dx.doi.org/10.1145/1044731.1044732}
  {\path{doi:10.1145/1044731.1044732}}.

\bibitem{ThorupZ06}
Mikkel Thorup and Uri Zwick.
\newblock Spanners and emulators with sublinear distance errors.
\newblock In {\em Proceedings of the Seventeenth Annual {ACM-SIAM} Symposium on
  Discrete Algorithms, {SODA} 2006, Miami, Florida, USA, January 22-26, 2006},
  pages 802--809, 2006.

\end{thebibliography}
\appendix
\section{Missing Proofs}
\subparagraph*{Proof of Lemma~\ref{lemma:dom-set-rand}}
Let $S_1$ be a uniformly random subset of $V$ of size $n_p$. We take $a$ to be the vertex 
of the maximum depth in $\outbfs(S_1)$. Also, $S_2$ is set to $\nin(a,n_q)$, which is computable 
in just  $O(m)$ time. By Lemma~\ref{lemma:probability}, with high probability, the set $\nin(a,n_q)$ 
contains a vertex of $S_1$, if not, then we re-sample $S_1$, and compute $a$ and $S_2$ again. The 
number of times we do re-sampling is $O(1)$ on expectation, thus the 
runtime for computing $(S_1,S_2)$ is $O(m)$ on expectation. 
Now for any positive integer $d$, if $S_1$ is not $d$-out-dominating (that is 
$\depth(\outbfs(S_1))\nleq d$), then $\inbfs(a,d)$ must have empty-intersection with $S_1$. 
This is possible only when $\inbfs(a,d)$ is a strict subset of $S_2=\nin(a,n_q)$, since the later 
set intersects with $S_1$. In such a case for any $v\in V$, 
$d_G(v,S_2)\leq d_G(v,\inbfs(a,d))\leq \max\{0,d_G(v,a)-d\}$.
On substituting $d=\lfloor{p~\inecc(a)}\rfloor$, we have that either $S_1$ is 
$\lfloor{p~\inecc(a)}\rfloor$-out-dominating or $S_2$ is  $\lceil q~\inecc(a)\rceil$-in%
-dominating.

\end{document}